\documentclass[11pt]{article}
\usepackage[T1]{fontenc}
\usepackage[utf8]{inputenc}
\usepackage{lmodern}
\usepackage{microtype}
\usepackage{fullpage}
\usepackage{graphicx}
\usepackage{amsmath,amssymb,amsthm,mathtools}
\usepackage{bm}
\usepackage{dsfont}
\usepackage{mathpazo}
\usepackage{booktabs,tabularx,longtable,array}
\usepackage{xcolor}
\definecolor{MainBlue}{HTML}{173B57}
\definecolor{AccentTeal}{HTML}{247B7B}
\definecolor{SoftBlue}{HTML}{F1F6FA}
\definecolor{SoftTeal}{HTML}{EFF8F7}
\definecolor{SoftGray}{HTML}{F5F6F7}
\definecolor{RuleGray}{HTML}{D8DEE4}
\definecolor{WarmGold}{HTML}{A66F18}
\usepackage{enumitem}
\usepackage{mdframed}
\usepackage{etoolbox}
\usepackage[numbers,sort&compress]{natbib}
\usepackage{hyperref}
\usepackage[nameinlink,noabbrev]{cleveref}
\hypersetup{
  colorlinks=true,
  linkcolor=blue!55!black,
  citecolor=blue!55!black,
  urlcolor=blue!55!black,
  pdfborder={0 0 0},
  hypertexnames=false,
  pdftitle={Quantum incapacity from observable shadows},
  pdfauthor={Chengkai Zhu and Xin Wang}
}
\allowdisplaybreaks
\date{}

\apptocmd{\thebibliography}{\raggedright}{}{\PackageWarning{template}{Bibliography layout hook could not be applied}}

\theoremstyle{plain}
\newtheorem{theorem}{Theorem}
\newtheorem{lemma}[theorem]{Lemma}
\newtheorem{proposition}[theorem]{Proposition}
\newtheorem{corollary}[theorem]{Corollary}
\theoremstyle{definition}
\newtheorem{definition}[theorem]{Definition}
\theoremstyle{remark}
\newtheorem{remark}[theorem]{Remark}

\mdfdefinestyle{theoremframe}{
backgroundcolor=SoftBlue,
linecolor=MainBlue,
linewidth=.65pt,
innerleftmargin=7pt,innerrightmargin=7pt,
innertopmargin=6pt,innerbottommargin=6pt,
needspace=6\baselineskip,
nobreak=true,
skipabove=9pt,skipbelow=9pt
}
\mdfdefinestyle{propositionframe}{
backgroundcolor=SoftTeal,
linecolor=AccentTeal,
linewidth=.6pt,
innerleftmargin=7pt,innerrightmargin=7pt,
innertopmargin=6pt,innerbottommargin=6pt,
needspace=6\baselineskip,
nobreak=true,
skipabove=8pt,skipbelow=8pt
}
\mdfdefinestyle{lemmaframe}{
backgroundcolor=SoftGray,
linecolor=RuleGray!60!black,
linewidth=.5pt,
innerleftmargin=7pt,innerrightmargin=7pt,
innertopmargin=5pt,innerbottommargin=5pt,
needspace=5\baselineskip,
nobreak=true,
skipabove=8pt,skipbelow=8pt
}
\surroundwithmdframed[style=theoremframe]{theorem}
\surroundwithmdframed[style=propositionframe]{proposition}
\surroundwithmdframed[style=propositionframe]{corollary}
\surroundwithmdframed[style=lemmaframe]{lemma}
\surroundwithmdframed[style=theoremframe]{definition}

\crefname{theorem}{theorem}{theorems}
\Crefname{theorem}{Theorem}{Theorems}
\crefname{proposition}{proposition}{propositions}
\Crefname{proposition}{Proposition}{Propositions}
\crefname{lemma}{lemma}{lemmas}
\Crefname{lemma}{Lemma}{Lemmas}
\crefname{corollary}{corollary}{corollaries}
\Crefname{corollary}{Corollary}{Corollaries}
\crefname{definition}{definition}{definitions}
\Crefname{definition}{Definition}{Definitions}
\crefname{remark}{remark}{remarks}
\Crefname{remark}{Remark}{Remarks}
\crefname{equation}{equation}{equations}
\Crefname{equation}{Equation}{Equations}

\newcommand{\id}{\operatorname{id}}
\newcommand{\idchan}{\operatorname{id}}
\newcommand{\idop}{\mathbf{1}}
\newcommand{\Lin}{\mathcal{L}}
\newcommand{\Dens}{\mathcal{D}}
\newcommand{\Tr}{\operatorname{Tr}}
\newcommand{\rank}{\operatorname{rank}}
\newcommand{\spec}{\operatorname{spec}}
\newcommand{\supp}{\operatorname{supp}}
\newcommand{\Var}{\operatorname{Var}}
\newcommand{\cA}{\mathcal{A}}
\newcommand{\cB}{\mathcal{B}}
\newcommand{\cC}{\mathcal{C}}
\newcommand{\cD}{\mathcal{D}}
\newcommand{\cH}{\mathcal{H}}
\newcommand{\cJ}{\mathcal{J}}
\newcommand{\cM}{\mathcal{M}}
\newcommand{\cN}{\mathcal{N}}
\newcommand{\cR}{\mathcal{R}}
\newcommand{\cT}{\mathcal{T}}
\newcommand{\cV}{\mathcal{V}}

\newcommand{\cZ}{\mathcal{Z}}
\newcommand{\HS}{\mathrm{HS}}
\newcommand{\Ic}{I_{\mathrm c}}
\newcommand{\ad}{\dagger}
\newcommand{\ket}[1]{|#1\rangle}
\newcommand{\bra}[1]{\langle#1|}
\newcommand{\ketbra}[2]{|#1\rangle\!\langle#2|}
\newcommand{\REord}{\succeq_{\mathrm{cRE}}}
\newcommand{\BKMord}{\succeq_{\mathrm{cBKM}}}

\newcommand{\Span}{\operatorname{span}}
\newcommand{\GammaB}{\Gamma_{BB'}}

\def\newtitle{Quantum incapacity from observable shadows}
\title{\bfseries \boldmath \newtitle}
\author{%
  Chengkai~Zhu$^{1}$, Xin~Wang$^{2,\ast}$\\[0.65em]
  {\small $^{1}$QudeLeap Research, Shanghai 200030, China.}\\
  {\small $^{2}$The Hong Kong University of Science and Technology (Guangzhou),
  Guangzhou 511453, China.}
}

\begin{document}

\addtocontents{toc}{\protect\setcounter{tocdepth}{-1}}

\maketitle

\renewcommand{\thefootnote}{\fnsymbol{footnote}}
\footnotetext[1]{\href{mailto:felixxinwang@hkust-gz.edu.cn}{felixxinwang@hkust-gz.edu.cn}}

\begin{abstract}
Quantum communication is ruled out when an eavesdropper can reconstruct the receiver's state. Yet states are operationally specified by measurement
statistics, raising a sharper question: can access to every expectation value
destroy communication? We show that it can. We call unbiased, variance-nonincreasing reconstruction under arbitrary reference extensions \emph{complete observable shadowing} and prove that it forces unassisted private
and quantum capacities to vanish. An explicit qutrit channel
realizes this mechanism while being neither antidegradable nor PPT, resolving
the long-standing questions of whether quantum capacity can vanish outside
those two classes and whether antidegradability is the only
route to zero private capacity. Notably, it superactivates every finite-dimensional channel with zero quantum but positive private capacity. Statistical access can therefore preclude standalone
communication without erasing latent coherent utility.
\end{abstract}

\section{Introduction}
A quantum state is determined by its observable expectations, but reproducing those expectations need not amount to physically preparing the state. Shadow tomography and classical shadows illustrate the distinction between estimating properties and reconstructing a complete state~\cite{Aaronson2020Shadow,HuangKuengPreskill2020}. For communication, a more
specific question arises: can the environment estimate every receiver observable at least as precisely as the receiver, even when it cannot physically reproduce
the receiver's output? We show that such statistical access can rule out both private and quantum transmission, provided that the comparison is stable under
arbitrary quantum references.

For a channel $\cN_{A\to B}$, a complementary channel $\cN^c_{A\to E}$ describes the information transferred to the environment. Antidegradability
means that $\cN=\cD\circ\cN^c$ for a physical channel $\cD_{E\to B}$.
Data processing then forces the unassisted private and quantum capacities,
$P(\cN)$ and $Q(\cN)$, to vanish. Channels with positive-partial-transpose
(PPT) Choi operators provide another standard mechanism for zero quantum capacity~\cite{HorodeckiEtAl2000,SmithSmolin2012}. These mechanisms left two long-standing problems unresolved: whether an explicit finite-dimensional
channel can have $Q=0$ while being neither PPT nor
antidegradable~\cite{SmithSmolin2012,SinghDatta2022}, and whether a finite-dimensional channel can have zero private capacity without being antidegradable~\cite{Buscemi2014}.

We address both questions by introducing \emph{complete observable shadowing}
and constructing an explicit qutrit channel that satisfies this criterion.
The criterion requires a fixed linear
correspondence from receiver observables to unbiased Hermitian environment
estimators whose variances are no larger. Both requirements hold under every
finite quantum reference extension. This condition concerns observable expectation values, or outcome statistics of protocols, whose estimation is a central objective in many quantum information-processing tasks~\cite{HuangKuengPreskill2020,PeruzzoEtAl2014,TemmeBravyiGambetta2017,ZhaoEtAl2023ShadowInformation,YuanEtAl2024VirtualDistillation,ZhaoWangChiribella2024}.

We construct an exact certificate for this condition and verify it for a qutrit channel that is neither PPT nor antidegradable. Its optimized private
and coherent information vanish at every blocklength, although its Choi state is one-copy distillable with two-way classical assistance. The same fixed qutrit channel superactivates every finite-dimensional channel with zero quantum but positive private capacity. We give an explicit PPT partner and, combining our theorem with a separate private-capacity superactivation result,
also obtain two zero-private-capacity factors with positive joint quantum capacity. These conclusions distinguish observable reconstruction from physical
simulation (Fig.~\ref{fig:overview}) and standalone capacity from activation.

\begin{figure}[!t]
\centering
\IfFileExists{mainfigure.pdf}{%
    \includegraphics[width=0.9\linewidth]{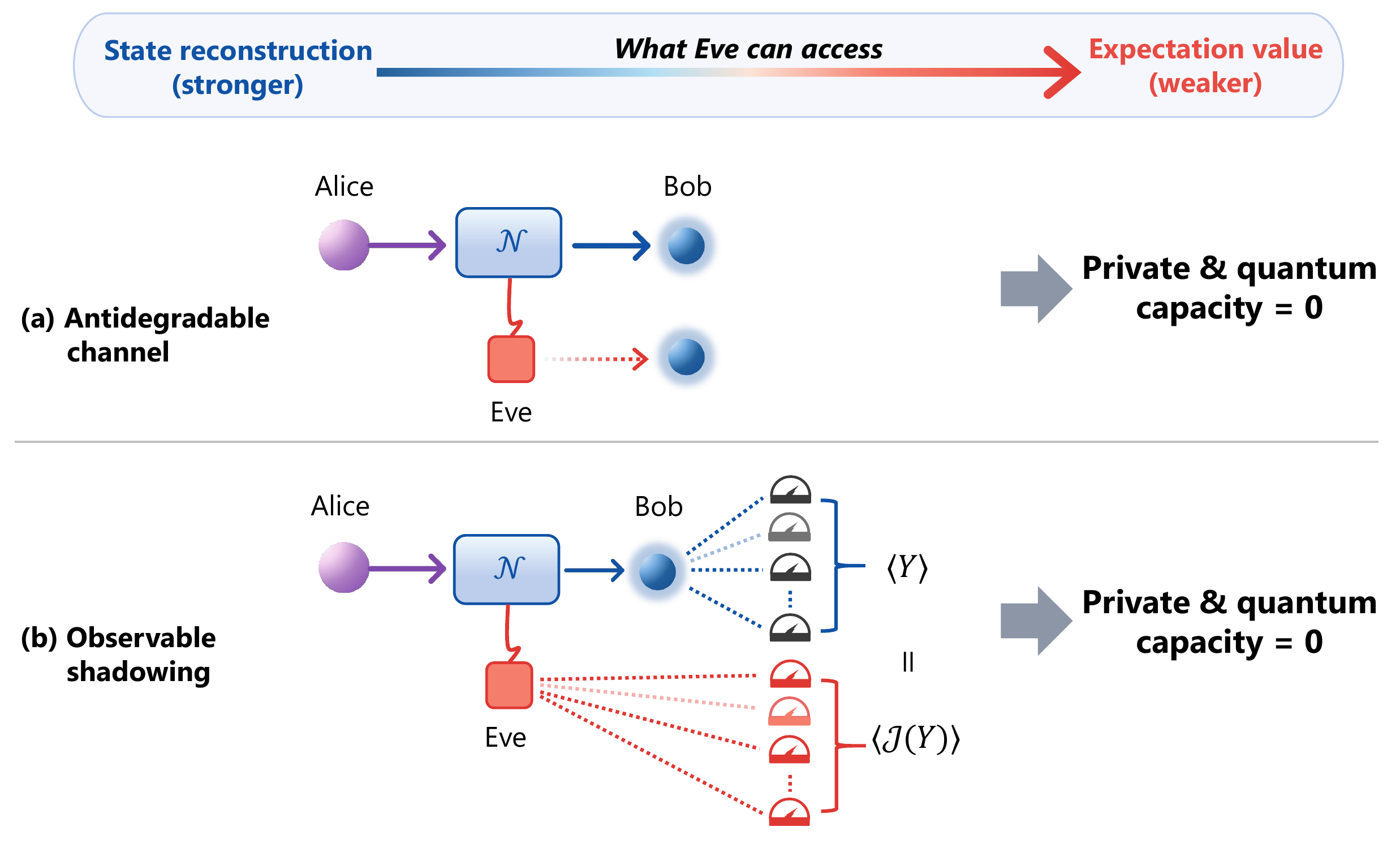}%
}{%
    \makebox[0.9\linewidth][c]{\rule{0pt}{5cm}}%
}
\caption{\textbf{State reconstruction versus complete observable shadowing.} Alice sends an input state $\rho$ through $\cN$ to Bob; Eve receives the complementary output $\cN^c(\rho)$. \textbf{(a)} Antidegradability allows Eve to reconstruct Bob's output through a
physical channel $\cD$, with $\cN=\cD\circ\cN^c$.
\textbf{(b)} Complete observable shadowing assigns each Hermitian receiver
observable $Y$ an environment estimator $\cJ(Y)$ through a fixed complex-linear,
adjoint-preserving map. For every input state, measuring $\cJ(Y)$ on Eve's
output reproduces the expectation of $Y$ on Bob's output with no larger
variance. Both properties remain valid under arbitrary quantum reference
extensions, without granting Eve access to the reference, and lead to zero private and quantum capacity.}
    \label{fig:overview}
\end{figure}

\section{Incapacity from complete observable shadowing}

Let $\cN_{A\to B}$ and $\cM_{A\to E}$ be channels with the same input. We say that $\cM$ provides \emph{complete observable shadowing} of $\cN$ if there is a complex-linear, adjoint-preserving map
$\cJ:\Lin(B)\to\Lin(E)$ such that
\begin{align}
 (\cJ_{B\to E})^{\dagger}\circ\cM_{A\to E}&=\cN_{A\to B},
 \label{eq:main-lift}\\
 \Var_{(\id_R\otimes\cM_{A\to E})(\rho_{RA})}
 \!\left[(\id_R\otimes\cJ_{B\to E})(Y_{RB})\right]
 &\leq\Var_{(\id_R\otimes\cN_{A\to B})(\rho_{RA})}(Y_{RB}).
 \label{eq:main-var}
\end{align}
The second condition holds for every finite reference $R$, every state
$\rho_{RA}$, and every Hermitian $Y_{RB}$, with
$\Var_\theta(Y):=\Tr(\theta Y^2)-[\Tr(\theta Y)]^2$.
The adjoint in Eq.~\eqref{eq:main-lift} is taken with respect to the
Hilbert--Schmidt inner product. That equation makes the estimator unbiased;
Eq.~\eqref{eq:main-var} compares the sample-mean variances obtained by measuring
$(\id_R\otimes\cJ_{B\to E})(Y_{RB})$ and $Y_{RB}$ on equally many independent copies of the respective reference--output states.

If $\cN=\cD\circ\cM$ for a physical channel $\cD$, then
$\cJ=\cD^{\dagger}$ satisfies both conditions by the Kadison--Schwarz inequality~\cite{Kadison1952,Choi1974Schwarz}. In particular, every antidegradable channel is completely shadowed by its complement. The converse fails: our qutrit example has a nonpositive shadow map and admits no physical reconstruction. Unbiasedness alone is insufficient; our proof also uses the variance bound and its stability under reference extensions.

One key bottleneck in understanding complete observable shadowing is the infinite family of reference-assisted variance inequalities in Eq.~\eqref{eq:main-var}. To address this issue, we develop a finite
algebraic certificate. Regard $\mathsf K_B:=\Lin(B)$ as a Hilbert space with the Hilbert--Schmidt
inner product, and let $L_\sigma(X)=\sigma X$. Define the \emph{outer defect}
$\Delta_{\cJ}:\Lin(A)\to\Lin(\mathsf K_B)$ by
\begin{equation}
 \Delta_{\cJ}(\tau):=L_{\cN(\tau)}-
 \cJ^{\dagger}\circ L_{\cM(\tau)}\circ\cJ .
 \label{eq:main-defect}
\end{equation}
Given Eq.~\eqref{eq:main-lift}, complete positivity of this outer map is
equivalent to Eq.~\eqref{eq:main-var}. It also implies domination of the
Bogoliubov--Kubo--Mori metric
\cite{Petz1994,Petz1996,LesniewskiRuskai1999}. Integrating the metric comparison~\cite{BelzigEtAl2025} gives
\begin{equation}
 \begin{aligned}
 &D\!\left((\id_R\otimes\cM_{A\to E})(\rho_{RA})
 \middle\|(\id_R\otimes\cM_{A\to E})(\sigma_{RA})\right)\\
 &\quad\geq
 D\!\left((\id_R\otimes\cN_{A\to B})(\rho_{RA})
 \middle\|(\id_R\otimes\cN_{A\to B})(\sigma_{RA})\right)
 \end{aligned}
 \label{eq:main-relative}
\end{equation}
for every reference and every pair satisfying
$\supp\rho\subseteq\supp\sigma$. The comparison is first proved for faithful inputs and full-rank identity
outputs. Output compression removes the restriction on the output spaces, and a common-state regularization extends the comparison to all support-compatible input pairs (Supplementary Materials).

When $\cM=\cN^c$, Eq.~\eqref{eq:main-relative} says that $\cN$ is completely
anti-less-noisy. The tensorization of this order gives the standard capacity
consequences~\cite{HircheRouzeFranca2022,Watanabe2012}:
\begin{equation}
 P^{(1)}(\cN^{\otimes n})=Q^{(1)}(\cN^{\otimes n})=0
 \quad(n\geq1),\qquad P(\cN)=Q(\cN)=0 .
 \label{eq:main-incapacity}
\end{equation}
Here, $P^{(1)}$ and $Q^{(1)}$ are the optimized private and coherent information.
At every blocklength, the environment's Holevo information is at least the
receiver's for every ensemble, so private information is nonpositive. The reference-assisted comparison likewise makes coherent information
nonpositive. A singleton ensemble and a pure input, respectively, attain zero. Definitions and conventions are given in Section~\ref{sec:main} of the
Supplementary Materials. The subsequent sections prove the certificate, output-support reduction, and capacity consequences.

\section{An exact qutrit separation}

Consider the qutrit channel
\begin{equation}
 \Lambda(X)=\frac12X+\frac14\left[\Tr(X)\idop-X^{\mathsf T}\right].
 \label{eq:main-lambda}
\end{equation}
This is the $x=1/2$ noisy Werner--Holevo channel, at a parameter value not
covered by the previously established antidegradable region
\cite{WernerHolevo2002,RoofehKarimipour2024}.
A four-dimensional complement is obtained from
$K_0=\idop/\sqrt2$ and $K_{\mu+1}=A_\mu/2$, where
$(A_\mu)_{jk}=\epsilon_{\mu jk}$, all indices run from $0$ to $2$, and
$\epsilon_{012}=1$. For $Y\in\Lin(\mathbb C^3)$, define
\begin{equation}
 \cJ(Y)=\frac13
 \begin{pmatrix}
  \Tr Y&\sqrt2\,\boldsymbol a(Y)^{\mathsf T}\\
  -\sqrt2\,\boldsymbol a(Y)&2\Tr(Y)\idop-Y-2Y^{\mathsf T}
 \end{pmatrix},\quad
 \boldsymbol a(Y)=\begin{pmatrix}Y_{12}-Y_{21}\\Y_{20}-Y_{02}\\Y_{01}-Y_{10}\end{pmatrix}.
 \label{eq:main-signed-lift}
\end{equation}
Direct contraction gives $\cJ^\dagger\circ\Lambda^c=\Lambda$. The map is
not positive: $\cJ(\ketbra00)=\operatorname{diag}(1,-1,2,2)/3$.
Nevertheless, its outer defect has the exact Kraus factorization
\begin{equation}
\Delta_{\cJ}(\tau)=\frac{2}{9}\sum_{\mu=0}^{2}G_\mu\tau G_\mu^\dagger,\quad
G_\mu\ket i=\ket i\ket\mu+\frac12\ket\mu\ket i
-\frac12\delta_{\mu i}\sum_j\ket j\ket j.
\label{eq:main-factor}
\end{equation}
Here $G_\mu$ maps the qutrit input space into the Hilbert--Schmidt space,
with the identification $\ketbra jk\leftrightarrow\ket j\ket k$.
This factorization proves complete observable shadowing and therefore
\begin{equation}
 \begin{split}
 P(\Lambda)&=Q(\Lambda)=0,\\
 P^{(1)}(\Lambda^{\otimes n})&=Q^{(1)}(\Lambda^{\otimes n})=0
 \qquad(n\geq1).
 \end{split}
 \label{eq:main-lambda-zero}
\end{equation}

The channel is neither PPT nor antidegradable. Let
$\Phi_3=\ketbra{\Phi_3}{\Phi_3}$ with
$\ket{\Phi_3}=3^{-1/2}\sum_i\ket{ii}$, and let
$P^-=(\idop-F)/2$ be the antisymmetric projector, where $F$ swaps the
qutrits. The normalized Choi state and its partial-transpose spectrum are
\begin{equation}
 \omega^\Lambda=\frac12\Phi_3+\frac16P^- ,\qquad
 \spec[(\omega^\Lambda)^{T_A}]
 =\{(1/4)^{\times5},0,(-1/12)^{\times3}\}.
 \label{eq:main-choi}
\end{equation}
The negative subspace contains a Schmidt-rank-two vector, proving one-copy
distillability~\cite{HorodeckiDistillation1998}. This is compatible with zero
unassisted capacity: the distillation protocol permits two-way classical
communication.
We further show that $\Lambda$ is not antidegradable by proving that the reconstruction conditions for three computational-basis inputs are incompatible with positivity and trace preservation of an antidegrading map (see Proposition~\ref{prop:non-antidegradable} in the Supplementary Materials).
The environment thus dominates receiver observables in the stated statistical sense without admitting a physical reconstruction.

\section{Private information enables quantum activation}

For every finite-dimensional channel $\cM$, the same fixed $\Lambda$ satisfies
\begin{equation}
 P(\cM)>0\quad\Longrightarrow\quad Q(\cM\otimes\Lambda)>0.
 \label{eq:main-universal-activation}
\end{equation}
For partners with $Q(\cM)=0<P(\cM)$, this is superactivation. To prove it, encode a qubit into the span of
$(\ket0+i\ket1)/\sqrt2$ and $\ket2$, and apply a fixed receiver decoder.
The resulting channel $\cC$ obeys
\begin{align}
 \cC(\ketbra00)&=\ketbra00,&
 \cC(\ketbra11)&=\tfrac34\ketbra00+\tfrac14\ketbra11,\notag\\
 \cC(\ketbra01)&=\tfrac14\ketbra01,&
 \cC(\ketbra10)&=\tfrac14\ketbra10.
 \label{eq:main-restriction}
\end{align}
On the single-excitation subspace of $d$ copies, this gives the exact
flagged dephased-erasure channel
\begin{equation}
 \begin{split}
 \cT_{q,\eta}^{(d)}(\rho)
 &=q\cZ_\eta(\rho)\oplus(1-q)\Tr(\rho)\ketbra{\mathsf e}{\mathsf e},\\
 \cZ_\eta&=\eta\id+(1-\eta)\Delta,
 \qquad q=\eta=\frac14,
 \end{split}
 \label{eq:main-erasure}
\end{equation}
where $\Delta$ dephases the excitation-location basis and $\ket{\mathsf e}$
is an orthogonal erasure flag. A population survives with probability $1/4$;
a coherence between two locations is multiplied by $1/16$.

Let a private-code ensemble for $\cM$ have classical label $X$, and put
$a=H(X|E)$ and $b=H(X|B)$. Here $H(X|E)$ and $H(X|B)$ are conditional
von Neumann entropies of the corresponding classical--quantum states.
Purifying the signal states on mutually orthogonal auxiliary subspaces gives
an input to $\cM\otimes\cT_{q,\eta}^{(d)}$ satisfying
\begin{equation}
 I_{\mathrm c}\geq q\eta a-qh_2(\eta)-(1-q)b,
 \label{eq:main-conversion}
\end{equation}
where $h_2$ is binary entropy. For positive-rate private codes with $K$
equiprobable messages, secrecy gives $a/\log K\to1$ and reliability gives
$b/\log K\to0$. The penalty $qh_2(\eta)$ is independent of $d$. At
$q=\eta=1/4$, the positive term is $a/16$, so some finite block has positive
coherent information. Isometric encoding and receiver postprocessing realize
$\cT_{1/4,1/4}^{(d)}$ using tensor powers of $\Lambda$. Padding unused channel
uses with fixed pure inputs, then repeating that finite block, proves
Eq.~\eqref{eq:main-universal-activation}.

An explicit partner is the four-dimensional channel $\cM_H$ whose normalized
Choi state is a low-dimensional PPT key-distillable state
\cite{HorodeckiEtAl2008}. It has $Q(\cM_H)=0$, while a specified binary
input ensemble gives
\begin{equation}
 P^{(1)}(\cM_H)\geq1-h_2(2-\sqrt2)
 =0.0213399156\ldots>0.
 \label{eq:main-helper}
\end{equation}
Thus $Q(\cM_H\otimes\Lambda)>0$.
Privacy-based superactivation was established using an antidegradable erasure
factor~\cite{SmithYard2008}; related amplification bounds use channels induced
by private states~\cite{WuWang2025}. Here, the activating factor is a fixed
qutrit channel outside the PPT and antidegradable classes, and the theorem applies to every finite-dimensional positive-private-capacity partner.

\section{Discussion}

Complete observable shadowing relaxes physical reconstruction while retaining an information comparison strong enough to exclude unassisted communication. The condition applies observable by observable: it requires neither a
simultaneous classical record nor a simulation of all receiver measurement outcome distributions.
Variance domination controls a second moment, not every finite-confidence
performance criterion. It is sufficient for zero capacity, but is not claimed
to characterize all zero-capacity channels. The activation proof is also
existential: the required private-code block and auxiliary dimension can be
large, so the resulting positive rate need not be practical.

The tensorization property clarifies which partners are excluded. If another
channel $\cA$ also satisfies $\cA^c\succeq_{\mathrm{cRE}}\cA$, then
$P(\cA\otimes\Lambda)=0$. In particular, an antidegradable channel cannot
activate $\Lambda$ by itself. This is a stability property of the complete
order, not of zero private capacity in general.

After the first version of this work, recent work establishes private-capacity superactivation using two zero-private-capacity factors~\cite{ZhuWang2026Private,UthirakalyaniEtAl2026}. In particular, Ref.~\cite{ZhuWang2026Private} constructs a four-dimensional
channel $\cN_0$ such that $P(\cN_0)=0$ but
$P(\cN_0\otimes\mathcal E_{2,1/2})>0$, where $\mathcal E_{2,1/2}$ is the qubit half-erasure channel. Combining that result with our theorem gives
\begin{equation}
 \begin{split}
 \cA_0&:=\cN_0,\qquad
 \cB_0:=\mathcal E_{2,1/2}\otimes\Lambda,\\
 P(\cA_0)&=P(\cB_0)=0,\qquad
 Q(\cA_0\otimes\cB_0)>0.
 \end{split}
 \label{eq:main-zero-private-quantum}
\end{equation}
The zero private capacity of $\cB_0$ follows from complete-order tensorization;
positive joint quantum capacity follows by applying
Eq.~\eqref{eq:main-universal-activation} to
$\cN_0\otimes\mathcal E_{2,1/2}$. This corollary combines our theorem with the
separate construction of Ref.~\cite{ZhuWang2026Private}, rather than the PPT
helper above, and answers the question posed by Smith and Smolin of whether two zero-private-capacity channels can jointly transmit quantum
information~\cite{Smith2009b}.

Approximate forms of the shadowing criterion may lead to quantitative capacity
bounds. The exact example already shows that observable-level statistical
dominance can coexist with activation, even when physical reconstruction is
impossible.

\section*{AI statement}
The authors formulated the research question, selected the candidate channels under investigation, and directed the research. Large language models (LLMs) suggested using an adjoint-preserving, generally nonpositive observable map $\cJ$ satisfying
$\cJ^{\ad}\circ\Lambda^c=\Lambda$ in place of a physical antidegrading map, as a starting point for establishing complete anti-less-noisy domination.
LLMs were also used to develop the proof of quantum activation with $\Lambda$ and various computational checks. All proofs have been verified and refined by the authors, who take full responsibility for the accuracy and integrity of the work.

\section*{Acknowledgments}
This work was partially supported by the National Natural Science Foundation of China (Grant Nos.~92576114, 12447107) and the Guangdong Provincial Quantum Science Strategic Initiative (Grant Nos.~GDZX2403008, GDZX2503001, and GDZX2403001).



\newpage


\renewcommand{\thefigure}{S\arabic{figure}}
\renewcommand{\thetable}{S\arabic{table}}
\renewcommand{\theequation}{S\arabic{equation}}
\renewcommand{\thepage}{S\arabic{page}}
\renewcommand{\thesection}{S\arabic{section}}
\renewcommand{\thetheorem}{S\arabic{theorem}}
\setcounter{figure}{0}
\setcounter{table}{0}
\setcounter{equation}{0}
\setcounter{page}{1}
\setcounter{section}{0}
\setcounter{theorem}{0}

\begin{center}
\section*{Supplementary Materials for\\ \newtitle}
Chengkai~Zhu and Xin~Wang
\end{center}

\addtocontents{toc}{\protect\setcounter{tocdepth}{2}}

\begingroup
\renewcommand{\baselinestretch}{1.05}\small
\tableofcontents
\endgroup

\section{Preliminaries}
\label{sec:main}

All Hilbert spaces are finite dimensional. All logarithms are base two unless explicitly written as $\ln$, which denotes the natural logarithm. Quantum channels are completely positive and trace preserving. The capacities $P$ and $Q$ are unassisted: no preshared key, entanglement, or auxiliary communication
is supplied. Assisted capacities are explicitly indicated. One use of a
product channel consumes one use of each factor.

For systems $A$ and $B$, let
$\Lin(\cH_A,\cH_B)$ denote the space of linear operators from $\cH_A$ to
$\cH_B$, and set $\Lin(\cH_A):=\Lin(\cH_A,\cH_A)$.  The Hilbert--Schmidt inner product on
$\Lin(\cH_A)$ is $\langle X_A,Y_A\rangle_{\HS}:=\Tr(X_A^{\ad}Y_A)$, and
$\Dens(\cH_A)$ denotes the set of density operators on $\cH_A$.  We use
$\idop_A$ to denote the identity operator on $\cH_A$. For \(n\ge1\), write \(A^n:= A_1\cdots A_n\) and $\cH_{A^n}:=\cH_A^{\otimes n}$. Let $\cN_{A\to B}:\Lin(\cH_A)\to\Lin(\cH_B)$ be a quantum channel. Fix an isometric extension $V^{\cN}_{A\to BE}\in\Lin(\cH_A,\cH_B\otimes\cH_E)$, so that
\begin{equation}
\cN_{A\to B}(X_A)
=\Tr_E\!\left[
V^{\cN}X_A(V^{\cN})^{\ad}
\right],
\qquad
\cN^c_{A\to E}(X_A)
=\Tr_B\!\left[
V^{\cN}X_A(V^{\cN})^{\ad}
\right],
\label{eq:channel-complement}
\end{equation}
where $\cN^c_{A\to E}$ is called the complementary channel of $\cN_{A\to B}$. For a state $\rho_C$, the von Neumann entropy is denoted by $S(C)_\rho:=-\Tr(\rho_C\log\rho_C)$. For a bipartite state $\rho_{CD}$, the quantum mutual information and
conditional entropy are
\[
I(C;D)_\rho
:= S(C)_\rho+S(D)_\rho-S(CD)_\rho,
\qquad
S(C|D)_\rho
:= S(CD)_\rho-S(D)_\rho.
\]

For a classical register $X$, we also write
$H(X|B):=S(X|B)$; $H(X)$ is Shannon entropy. Binary entropy is
$h_2(t):=-t\log t-(1-t)\log(1-t)$, with $0\log0=0$.
All complements below are full complementary channels, not reduced
subsystems of the environment. For a system $F$, write
$d_F:=\dim\cH_F$ and $E^F_{ij}:=\ketbra ij_F$ in its fixed basis; the
system label is omitted when clear. For any linear map $\Phi$ and reference
$R$, write $\Phi_R:=\id_R\otimes\Phi$. The same convention applies to
composite references, such as $\Phi_{QR}:=\id_{QR}\otimes\Phi$.

\paragraph{Quantum communication.}
The unassisted quantum capacity $Q(\cN)$ is the supremum of the asymptotic
rates at which quantum information can be transmitted reliably over
independent uses of $\cN$.  For
$\rho_A\in\Dens(\cH_A)$, let $\psi_{RA}$ be a purification of $\rho_A$ and set $\omega_{RBE}
:=(\idop_R\otimes V^{\cN}_{A\to BE})
\psi_{RA}(\idop_R\otimes V^{\cN}_{A\to BE})^{\ad}$.
The coherent information of $\rho_A$ through $\cN$ is defined by
\begin{equation}
\Ic(\rho_A,\cN)
:= -S(R|B)_\omega
=S(B)_\omega-S(E)_\omega.
\label{eq:coherent-info}
\end{equation}
Coherent information cannot increase under receiver postprocessing: for every
quantum channel $\cD_{B\to B'}$,
\begin{equation}
\Ic(\rho_A,\cD\circ\cN)\leq\Ic(\rho_A,\cN).
\label{eq:receiver-dpi}
\end{equation}
Indeed, let $U_{B\to B'F}$ be an isometric extension of $\cD$ and set
$\tau_{RB'F}:=(\idop_R\otimes U)\omega_{RB}(\idop_R\otimes U)^{\ad}$.
Isometric invariance and strong subadditivity give
\[
\Ic(\rho_A,\cD\circ\cN)
=-S(R|B')_\tau
\leq -S(R|B'F)_\tau
=-S(R|B)_\omega
=\Ic(\rho_A,\cN),
\]
which proves Eq.~\eqref{eq:receiver-dpi}.
The optimized coherent information of one channel use is
\begin{equation}
Q^{(1)}(\cN)
:=
\max_{\rho_A\in\Dens(\cH_A)}
\Ic(\rho_A,\cN).
\label{eq:quantum-one}
\end{equation}
The quantum coding theorem identifies the quantum capacity as the regularized coherent
information~\cite{Lloyd1997,Shor2002,Devetak2005}
\begin{equation}
Q(\cN)
=
\lim_{n\to\infty}
\frac{1}{n}Q^{(1)}(\cN^{\otimes n}).
\label{eq:quantum-capacity-regularization}
\end{equation}
Regularization is necessary in general: entangled inputs across several channel uses can yield larger coherent information per use than every product input~\cite{ShorSmolin1996,DiVincenzoShorSmolin1998,SmithSmolin2007}. More strongly, no universal finite blocklength suffices even to determine whether the quantum capacity is positive~\cite{CubittEtAl2015}.

\paragraph{Private classical communication.}
The unassisted private capacity $P(\cN)$ is the supremum of the asymptotic
rates at which classical messages can be decoded reliably by the receiver
while becoming asymptotically independent of the environment.  For a finite
input ensemble
$\{p(u),\rho_A^u\}_{u\in\mathcal U}$, define
\begin{equation}
\omega_{UBE}
:=
\sum_{u\in\mathcal U}
p(u)\ketbra{u}{u}_U
\otimes
V^{\cN}_{A\to BE}\rho_A^u
(V^{\cN}_{A\to BE})^{\ad}.
\label{eq:cq-output-state}
\end{equation}
The corresponding private information is the advantage of the receiver over
the environment.  Optimizing this advantage over finite input ensembles gives
\begin{equation}
P^{(1)}(\cN)
:=
\sup_{\{p(u),\rho_A^u\}}
\left[
I(U;B)_\omega-I(U;E)_\omega
\right].
\label{eq:private-one}
\end{equation}
The private coding theorem identifies the private capacity with the
regularized private information~\cite{Devetak2005}.
\begin{equation}
P(\cN)
=
\lim_{n\to\infty}
\frac{1}{n}P^{(1)}(\cN^{\otimes n})
=
\sup_{n\ge1}
\frac{1}{n}P^{(1)}(\cN^{\otimes n}).
\label{eq:private-capacity-regularization}
\end{equation}
Again, the existence of the limit follows from superadditivity and
finite dimensionality.  Private information can also remain strictly
superadditive over arbitrarily many channel uses~\cite{ElkoussStrelchuk2015}.

\section{Complete observable shadowing and less-noisy order}
\label{sec:general-framework}

Antidegradability establishes zero capacity through a physical simulation of the receiver's output from the environment's complementary output. For channels outside the antidegradable class, however, requiring a complete reconstruction of the output state may be unnecessarily strong. This suggests seeking a weaker, information-theoretic comparison. Even after adjoining an arbitrary reference system, the environment's output should be at least as distinguishable as the receiver's for every pair of input hypotheses, as measured by Umegaki relative entropy. The complete less-noisy order provides a natural formulation of this idea. Its tensorization property further suggests how a single-use comparison might extend to all blocklengths.

The less-noisy viewpoint originates in classical channel comparison and was brought into quantum capacity theory by Watanabe~\cite{Watanabe2012}.  The reference-complete relative-entropy formulation and its tensorization were developed systematically by Hirche, Rouz\'e, and
Stilck Fran\c{c}a~\cite{HircheRouzeFranca2022}. We recall only the forms needed here. The new ingredient is the complete observable-shadowing criterion in~Theorem~\ref{thm:criterion}, together with its equivalent CP-defect certificate. It replaces physical postprocessing by an adjoint-preserving, generally nonpositive shadow map that is unbiased and variance contracting under arbitrary references.

\subsection{Complete less-noisy order}

For states $\rho$ and $\sigma$, the Umegaki relative entropy is defined by
\begin{equation}
D(\rho\Vert\sigma)
:=
\begin{cases}
    \Tr[\rho(\log\rho-\log\sigma)],
    &\supp\rho\subseteq\supp\sigma,\\
    +\infty,&\text{otherwise}.
\end{cases}
\label{eq:relative-entropy}
\end{equation}
Let
$\cM_{A\to E}:\Lin(\cH_A)\to\Lin(\cH_E)$ and
$\cN_{A\to B}:\Lin(\cH_A)\to\Lin(\cH_B)$ be channels with the same input.
We write
\begin{equation}
\cM\REord\cN
\label{eq:crel-symbol}
\end{equation}
if, for every finite-dimensional reference system $R$ and all
$\rho_{RA},\sigma_{RA}\in\Dens(\cH_R\otimes\cH_A)$ satisfying
$\supp\rho_{RA}\subseteq\supp\sigma_{RA}$,
\begin{equation}
D(\cM_R(\rho)\Vert\cM_R(\sigma))
\ge
D(\cN_R(\rho)\Vert\cN_R(\sigma)),
\qquad
\cM_R:=\id_R\otimes\cM,\quad
\cN_R:=\id_R\otimes\cN.
\label{eq:crel-definition}
\end{equation}
In finite dimensions the support condition implies
$\rho_{RA}\le c\sigma_{RA}$ for some $c<\infty$.  Since $\cM$ and
$\cN$ are completely positive, their reference amplifications
$\cM_R$ and $\cN_R$ are positive.  Consequently,
\[
\cM_R(\rho)\le c\cM_R(\sigma),
\qquad
\cN_R(\rho)\le c\cN_R(\sigma),
\]
so both relative entropies in
Eq.~\eqref{eq:crel-definition} are finite.
Hirche, Rouz\'e, and Stilck Fran\c{c}a formulate the complete less-noisy preorder through a reference-assisted mutual-information inequality.
Applying Proposition~2.3 of Ref.~\cite{HircheRouzeFranca2022} to
$\id_R\otimes\cM$ and $\id_R\otimes\cN$ for each $R$, and then quantifying
over all finite $R$, gives Eq.~\eqref{eq:crel-definition}. The subscript
$\mathrm{cRE}$ emphasizes this relative-entropy characterization.
Lemma~3.1 of Ref.~\cite{HircheRouzeFranca2022} establishes the two-pair
tensorization rule
\begin{equation}
\cM_1\REord\cN_1,\quad \cM_2\REord\cN_2
\quad\Longrightarrow\quad
(\cM_1\otimes\cM_2)\REord(\cN_1\otimes\cN_2).
\label{eq:hrsf-two-pair}
\end{equation}
To see this directly, fix a reference $R$ and a support-compatible pair on
$RA_1A_2$. For $i=1,2$, let $E_i$ and $B_i$ be the output systems of
$\cM_i$ and $\cN_i$, respectively. After applying $\cM_2$, compare
$\cM_1$ with $\cN_1$ using $RE_2$ as the reference. Next, after applying
$\cN_1$, compare $\cM_2$ with $\cN_2$ using $RB_1$ as the reference.
The two relative-entropy inequalities yield Eq.~\eqref{eq:hrsf-two-pair}.
All intermediate pairs retain support inclusion because the applied maps are completely positive.
Iterating this implication gives
\begin{equation}
\cM\REord\cN
\quad\Longrightarrow\quad
\cM^{\otimes n}\REord\cN^{\otimes n}
\qquad(n\ge1).
\label{eq:tensorized-order}
\end{equation}

The following capacity consequence is standard for anti-less-noisy
channels~\cite{Watanabe2012,HircheRouzeFranca2022}. We include the proof to
show explicitly how the relative-entropy comparison bounds both private and
coherent information at every blocklength.

\begin{corollary}[Standard anti-less-noisy capacity consequence]
\normalfont\par\noindent
\label{cor:anti-less-noisy-capacity}
Let $\cN_{A\to B}$ be a channel with complementary channel
$\cN^c_{A\to E}$.  If
\begin{equation}
    \cN^c\REord\cN,
    \label{eq:complement-less-noisy}
\end{equation}
then, for every $n\ge1$,
\begin{equation}
    P^{(1)}(\cN^{\otimes n})
    =Q^{(1)}(\cN^{\otimes n})=0,
    \qquad
    P(\cN)=Q(\cN)=0.
    \label{eq:general-all-block-zero}
\end{equation}
\end{corollary}

\begin{proof}
Eq.~\eqref{eq:tensorized-order} gives
$(\cN^c)^{\otimes n}\REord\cN^{\otimes n}$ for every $n\ge1$.
For an ensemble $\{p(u),\rho^u_{A^n}\}$, discard any zero-probability
terms and set $\bar\rho_{A^n}:=\sum_u p(u)\rho^u_{A^n}$.  Since
$\supp\rho^u\subseteq\supp\bar\rho$, the relative-entropy
characterization of Holevo information gives
\begin{align}
    I(U;E^n)-I(U;B^n)=\sum_u p(u)\Bigl[
    D\!\left((\cN^c)^{\otimes n}(\rho^u)
    \middle\Vert(\cN^c)^{\otimes n}(\bar\rho)\right)-
    D\!\left(\cN^{\otimes n}(\rho^u)
    \middle\Vert\cN^{\otimes n}(\bar\rho)\right)
    \Bigr]\ge0.
    \label{eq:private-from-order}
\end{align}
Thus every private-information difference is nonpositive; a singleton ensemble attains zero.

For the coherent information, let $\psi_{RA^n}$ purify
$\rho_{A^n}$.  The Schmidt decomposition gives
$\supp\psi_{RA^n}\subseteq
\supp(\psi_R\otimes\rho_{A^n})$, so the reference-complete order applies
to $\bigl(\psi_{RA^n},\psi_R\otimes\rho_{A^n}\bigr)$.  This gives
$I(R;E^n)\ge I(R;B^n)$.  The joint state on $RB^nE^n$ is pure, and hence
\begin{equation}
    2\Ic(\rho_{A^n},\cN^{\otimes n})
    =I(R;B^n)-I(R;E^n)\le0.
    \label{eq:coherent-from-order}
\end{equation}
A pure input attains coherent information zero.  Therefore, both optimized one-shot quantities vanish for every $n$, and their regularizations vanish as well. The same conclusion follows from Proposition~3.2 and the zero-parameter case of Theorem~4.13 of Ref.~\cite{HircheRouzeFranca2022}.
\end{proof}

\subsection{The BKM-to-relative-entropy bridge}

The complete less-noisy order compares relative entropies globally.  We
instead verify their Hessians locally at every reference amplification and
then use the standard affine-path integration argument. For a faithful state $\sigma$ and Hermitian $X$, set
\begin{align}
\Gamma_\sigma(Y)
&:=
\int_0^1\sigma^sY\sigma^{1-s}\,ds,
\label{eq:kubo-mori-map}\\
g_\sigma(X)
&:=
\langle X,\Gamma_\sigma^{-1}(X)\rangle_{\HS}=
\int_0^\infty
\Tr\!\left[
X(\sigma+s\idop)^{-1}
X(\sigma+s\idop)^{-1}
\right]ds.
\label{eq:dual-bkm-form}
\end{align}
The form $g_\sigma$ is the BKM metric in the natural-logarithm
convention, a standard monotone quantum
metric~\cite{Petz1994,Petz1996,LesniewskiRuskai1999}. For base-two relative
entropy, the Hessian quadratic form on traceless Hermitian perturbations is
$g_\sigma/\ln 2$; the common positive factor does not affect metric comparisons.
Let $\cM(\idop_A)\succ0$ and $\cN(\idop_A)\succ0$.  We write
$\cM\BKMord\cN$ if, for every finite reference $R$, every faithful
$\rho_{RA}$, and every traceless Hermitian $X_{RA}$,
\begin{equation}
g_{\cM_R(\rho)}(\cM_R(X))
\ge
g_{\cN_R(\rho)}(\cN_R(X)).
\label{eq:cbkm-definition}
\end{equation}
In words, $\cM$ completely dominates $\cN$ at the level of the
relative-entropy Hessian. For faithful $\rho,\sigma$, set $X:=\rho-\sigma$ and
$\rho_t:=(1-t)\sigma+t\rho$.  For
$f(t):= D(\rho_t\Vert\sigma)$, one has
$f(0)=f'(0)=0$ and $f''(t)=g_{\rho_t}(X)/\ln 2$.  Taylor's theorem with integral
remainder gives the standard integral representation (see the proof of
Lemma~2.2 of Ref.~\cite{GaoRouze2022} and Eq.~(15) of
Ref.~\cite{BelzigEtAl2025})
\begin{equation}
D(\rho\Vert\sigma)=\frac{1}{\ln 2}\int_0^1(1-t)g_{\rho_t}(X)\,dt.
\label{eq:relative-entropy-bkm}
\end{equation}
The following lemma is the lower-bound direction of the BKM-to-relative-entropy
comparison in Lemma~4.1 of Ref.~\cite{BelzigEtAl2025}. We include the
boundary regularization needed for support-compatible states.

\begin{lemma}[BKM-to-relative-entropy bridge]
\normalfont\par\noindent
\label{lem:complete-bkm-to-cre}
Let $\cM_{A\to E}$ and $\cN_{A\to B}$ be finite-dimensional quantum
channels such that
$\cM(\idop_A)\succ0$ and $\cN(\idop_A)\succ0$.  Then
\begin{equation}
    \cM\BKMord\cN \quad\Longrightarrow\quad
    \cM\REord\cN.
    \label{eq:local-global-conclusion}
\end{equation}
\end{lemma}

\begin{proof}
Fix a finite reference $R$ and states $\rho_{RA},\sigma_{RA}$ satisfying
$\supp\rho\subseteq\supp\sigma$.  Let
$\tau_{RA}=\idop_{RA}/(d_Rd_A)$ and, for $0<\varepsilon<1$, put
\[
\rho_\varepsilon=(1-\varepsilon)\rho+\varepsilon\tau,
\qquad
\sigma_\varepsilon=(1-\varepsilon)\sigma+\varepsilon\tau.
\]
Both states are faithful.  Writing
$\rho_{\varepsilon,t}=(1-t)\sigma_\varepsilon+t\rho_\varepsilon$,
one has $\rho_{\varepsilon,t}\succeq\varepsilon\tau$ and hence, for
$\Phi\in\{\cM,\cN\}$,
\[
\Phi_R(\rho_{\varepsilon,t})
\succeq\frac{\varepsilon}{d_Rd_A}
\idop_R\otimes\Phi(\idop_A)\succ0.
\]
Thus every channel output along the segment is faithful, so
Eq.~\eqref{eq:relative-entropy-bkm} and the complete BKM
order give
\begin{align}
    D\!\left(\cM_R(\rho_\varepsilon)
    \middle\Vert\cM_R(\sigma_\varepsilon)\right)
    \ge
    D\!\left(\cN_R(\rho_\varepsilon)
    \middle\Vert\cN_R(\sigma_\varepsilon)\right).
    \label{eq:bkm-regularized-re}
\end{align}
For either channel $\Phi\in\{\cM,\cN\}$, complete positivity preserves
support inclusion: $\rho\le c\sigma$ for some finite $c$, hence
$\Phi_R(\rho)\le c\Phi_R(\sigma)$.  The limiting relative entropy is
therefore finite.  Lower semicontinuity gives
\[
D\!\left(\Phi_R(\rho)\middle\Vert\Phi_R(\sigma)\right)
\le
\liminf_{\varepsilon\downarrow0}
D\!\left(\Phi_R(\rho_\varepsilon)
\middle\Vert\Phi_R(\sigma_\varepsilon)\right),
\]
whereas joint convexity, using the same state $\Phi_R(\tau)$ in both
regularizations, gives
\[
D\!\left(\Phi_R(\rho_\varepsilon)
\middle\Vert\Phi_R(\sigma_\varepsilon)\right)
\le
(1-\varepsilon)
D\!\left(\Phi_R(\rho)\middle\Vert\Phi_R(\sigma)\right).
\]
Thus the regularized relative entropies converge to their support-compatible
limits.  Letting $\varepsilon\downarrow0$ in
Eq.~\eqref{eq:bkm-regularized-re} proves $\cM\REord\cN$.
\end{proof}

For a concrete channel pair, it therefore remains only to prove the local
order $\cM\BKMord\cN$.  The next subsection gives an algebraic certificate
for that purpose.

\subsection{Complete observable shadowing}

The criterion below is motivated by the variance contraction obeyed by physical postprocessing, while allowing the reconstruction map to be nonpositive.  For a state $\theta$ and a Hermitian observable $Y$, write
\begin{equation}
\operatorname{Var}_{\theta}(Y)
:=
\Tr(\theta Y^2)-\Tr(\theta Y)^2.
\end{equation}
Suppose that $\cN=\cD\circ\cM$ for a quantum channel
$\cD:\Lin(\cH_E)\to\Lin(\cH_B)$ and put $\cJ=\cD^{\ad}$.  Then
$\cJ^{\ad}\circ\cM=\cN$, so measuring $\cJ(Y_B)$ on the output of $\cM$ reconstructs the expectation of $Y_B$ at the output of $\cN$. Because $\cJ=\cD^{\ad}$ is unital and completely positive, the Kadison--Schwarz inequality gives~\cite{Kadison1952,Choi1974Schwarz}
\begin{equation}
    \cJ(Y_B^2)\succeq \cJ(Y_B)^2
    \quad\Longrightarrow\quad
    \operatorname{Var}_{\cM(\theta)}\!\left(\cJ(Y_B)\right)
    \leq \operatorname{Var}_{\cN(\theta)}(Y_B).
\end{equation}
Such reconstruction of observable expectations with variance contraction holds
directly for antidegradable channels, by taking $\cM:=\cN^c$. We retain both
statistical guarantees but relax the assumption that $\cJ$ is positive. This
gives the following definition.

\begin{definition}[Complete observable shadowing]
\normalfont\par\noindent
\label{def:complete-variance-lift}
Let $\cN_{A\to B}$ and $\cM_{A\to E}$ be quantum channels. We say that
$\cM$ provides \emph{complete observable shadowing} of $\cN$ if there exists a
complex-linear, adjoint-preserving map
$\cJ:\Lin(\cH_B)\to\Lin(\cH_E)$ such that
\begin{equation}
    \cJ^{\ad}\circ\cM=\cN
    \label{eq:criterion-lift}
\end{equation}
and, for every finite reference $R$, every state $\rho_{RA}$, and every
Hermitian $Y_{RB}$,
\begin{equation}
    \operatorname{Var}_{\cM_R(\rho)}
    \!\left(\cJ_R(Y)\right)
    \le
    \operatorname{Var}_{\cN_R(\rho)}(Y),
    \qquad
    \cJ_R:=\id_R\otimes\cJ.
    \label{eq:complete-variance-domination}
\end{equation}
We call any such $\cJ$ a \emph{shadow map}. The first condition makes the
observable estimator unbiased, and the second
makes it variance contracting. The adjective \emph{complete} refers to the
quantification over arbitrary finite references. The shadow map is not assumed
positive; nevertheless, adjoint preservation makes $\cJ(Y)$ Hermitian whenever
$Y$ is Hermitian.
\end{definition}

Nonpositive linear maps also appear in quantum statistical comparison as
quantum statistical morphisms (Definition~5 of Ref.~\cite{Buscemi2016}). Such a
morphism need not be positive on the full operator space, but on the specified
family of source states it must reproduce the statistics of each target
positive-operator-valued measure (POVM) by a suitable source POVM.  In the present notation, we do not assume that the Schr\"odinger-picture reconstruction
$\cJ^{\ad}$ is a statistical morphism, and for a receiver POVM $\{Q_y\}_y$, the generally nonpositive operators $\{\cJ(Q_y)\}_y$ need not form an environment POVM. Instead, the shadow map provides a fixed observable correspondence $Y_B\longmapsto\cJ(Y_B)$ that reconstructs the expectation of every receiver observable and satisfies a variance comparison stable under arbitrary quantum references.

We next give a finite algebraic certificate for this physically motivated
statistical condition.  For
any system $F$, write $\mathsf K_F:=\Lin(\cH_F)$ when the operator
space is regarded as a Hilbert space with the Hilbert--Schmidt inner product.
Thus elements of $\Lin(\mathsf K_F)$ are superoperators.  For
$\theta_F\in\Lin(\cH_F)$, let
\begin{equation}
L_\theta(Y):=\theta Y,
\qquad
R_\theta(Y):= Y\theta,
\qquad Y\in\mathsf K_F.
\label{eq:left-right}
\end{equation}
Superoperator inequalities below refer to the L\"owner order induced by the Hilbert--Schmidt inner product.  Given $\cJ$, define the \emph{outer defect} of $\cN$ through $\cM$ by
\begin{equation}
\begin{aligned}
    \Delta_{\cJ}:\Lin(\cH_A)&\longrightarrow\Lin(\mathsf K_B),\\
    \bigl[\Delta_{\cJ}(\tau_A)\bigr](Y_B)
    &:=
    \cN(\tau_A)Y_B
    -\cJ^{\ad}\!\left(\cM(\tau_A)\cJ(Y_B)\right),
\end{aligned}
\label{eq:criterion-defect}
\end{equation}
or equivalently
\begin{equation}
\Delta_{\cJ}(\tau_A)
=
L_{\cN(\tau_A)}
-\cJ^{\ad}\circ L_{\cM(\tau_A)}\circ\cJ.
\label{eq:criterion-defect-operator}
\end{equation}

\begin{proposition}[Variance and the CP defect]
\normalfont\par\noindent
\label{prop:variance-defect-equivalence}
Given quantum channels $\cN_{A\to B}$ and $\cM_{A\to E}$, assume that $\cJ:\Lin(\cH_B)\to \Lin(\cH_E)$ is complex-linear, adjoint preserving, and satisfies $\cJ^\ad\circ \cM = \cN$.  The following are equivalent.
\begin{enumerate}
\renewcommand{\labelenumi}{\textup{(\roman{enumi})}}
    \item $\cJ$ is a shadow map witnessing complete observable shadowing of
    $\cN$ by $\cM$.
    \item For every finite reference $R$, every $\tau_{RA}\succeq0$, and every
    operator $Z_{RB}$, not necessarily Hermitian,
    \begin{equation}
        \begin{aligned}
            \Tr\!\left[Z^{\ad}\cN_R(\tau)Z\right]\ge
            \Tr\!\left[
            \cJ_R(Z)^{\ad}\cM_R(\tau)\cJ_R(Z)
            \right].
        \end{aligned}
        \label{eq:complete-second-moment}
    \end{equation}
    \item The outer map $\Delta_{\cJ}$ in
    Eq.~\eqref{eq:criterion-defect} is completely positive.
\end{enumerate}
\end{proposition}

\begin{proof}
We first prove the equivalence of \textup{(ii)} and \textup{(iii)}.  Choose a
basis of $R$ and write
\[
\tau=\sum_{r,s}|r\rangle\!\langle s|\otimes\tau_{rs},
\qquad
Z=\sum_{r,t}|r\rangle\!\langle t|\otimes Z_{rt}.
\]
For each column $t$, put
\[
|z_t\rangle
:=
\sum_r|r\rangle_R\otimes|Z_{rt}\rangle_{\mathsf K_B}
\in\cH_R\otimes\mathsf K_B.
\]
Expanding the traces gives
\begin{align}
    \Tr\!\left[Z^{\ad}\cN_R(\tau)Z\right]
    -\Tr\!\left[
    \cJ_R(Z)^{\ad}\cM_R(\tau)\cJ_R(Z)
    \right]&=\sum_{r,s,t}
    \left\langle
    Z_{rt},
    \bigl[\Delta_{\cJ}(\tau_{rs})\bigr](Z_{st})
    \right\rangle\notag\\
    &=\sum_t
    \left\langle
    z_t,(\id_R\otimes\Delta_{\cJ})(\tau)z_t
    \right\rangle.
    \label{eq:defect-amplification-identity}
\end{align}
Complete positivity of $\Delta_{\cJ}$ therefore implies
Eq.~\eqref{eq:complete-second-moment}.  Conversely, fix
$\tau_{RA}\succeq0$ and an arbitrary vector
\[
|z\rangle=\sum_r |r\rangle\otimes|Z_r\rangle
\in \mathcal H_R\otimes\mathsf K_B,
\]
and choose a basis vector $|t_0\rangle_R$. Define
\[
    Z_{RB}
    :=
    \sum_r |r\rangle\!\langle t_0|\otimes Z_r.
\]
The blocks of $Z_{RB}$ are $Z_{rt}=\delta_{t,t_0}Z_r$. Consequently, the vector associated with its $t$th block column is
\[
    |z_t\rangle=
    \sum_r |r\rangle\otimes|Z_{rt}\rangle
    =\delta_{t,t_0}
    \sum_r |r\rangle\otimes|Z_r\rangle
    =\delta_{t,t_0}|z\rangle.
\]
Substituting this into Eq.~\eqref{eq:defect-amplification-identity} gives
\begin{align}
    \Tr\!\left[Z^{\ad}\cN_R(\tau)Z\right]
    -\Tr\!\left[\cJ_R(Z)^{\ad}\cM_R(\tau)\cJ_R(Z)
    \right] =
    \sum_t
    \left\langle z_t,(\id_R\otimes\Delta_{\cJ})(\tau)z_t
    \right\rangle =
    \left\langle z,(\id_R\otimes\Delta_{\cJ})(\tau)z
    \right\rangle.
\end{align}
Therefore $\left\langle z,(\id_R\otimes\Delta_{\cJ})(\tau)z\right\rangle\ge 0$. Since $z$ was arbitrary, $(\id_R\otimes\Delta_{\cJ})(\tau)\succeq0$. Since the finite-dimensional reference system $R$ and $\tau_{RA}\succeq0$ were also arbitrary, $\Delta_{\cJ}$ is completely positive.

Condition \textup{(ii)} implies \textup{(i)} directly.  Indeed, for a state
$\rho_{RA}$ and a Hermitian $Y_{RB}$, adjoint preservation makes
$\cJ_R(Y)$ Hermitian, while Eq.~\eqref{eq:criterion-lift} gives
\begin{equation}
    \Tr\!\left[\cN_R(\rho)Y\right]
    =
    \Tr\!\left[\cM_R(\rho)\cJ_R(Y)\right].
    \label{eq:complete-equal-means}
\end{equation}
Applying Eq.~\eqref{eq:complete-second-moment} with $Z=Y$ and subtracting
the common squared mean proves
Eq.~\eqref{eq:complete-variance-domination}.

It remains to prove that \textup{(i)} implies \textup{(ii)}.  Let
$\tau_{RA}\succeq0$ be nonzero, put $t=\Tr\tau$, and let $Q$ be a qubit.
For arbitrary $Z_{RB}$, define
\begin{align}
    \widehat\rho_{QRA}
    &:=
    |0\rangle\!\langle0|_Q\otimes\frac{\tau_{RA}}{t},
    \notag\\
    \widehat Y_{QRB}
    &:=
    |0\rangle\!\langle1|_Q\otimes Z
    +|1\rangle\!\langle0|_Q\otimes Z^{\ad}.
    \label{eq:hermitian-dilation}
\end{align}
The observable $\widehat Y$ is Hermitian.  Both its receiver mean and the
mean of
$(\id_Q\otimes\cJ_R)(\widehat Y)$ vanish because the two output states are
supported on the $|0\rangle\!\langle0|_Q$ block.  Moreover,
\begin{align}
    \operatorname{Var}_{\cN_{QR}(\widehat\rho)}(\widehat Y)
    &=
    \frac1t\Tr\!\left[Z^{\ad}\cN_R(\tau)Z\right],
    \notag\\
    \operatorname{Var}_{\cM_{QR}(\widehat\rho)}
    \!\left((\id_Q\otimes\cJ_R)(\widehat Y)\right)
    &=
    \frac1t\Tr\!\left[
    \cJ_R(Z)^{\ad}\cM_R(\tau)\cJ_R(Z)
    \right].
    \label{eq:dilation-second-moments}
\end{align}
Applying complete observable shadowing to the enlarged reference
$Q\otimes R$ gives
Eq.~\eqref{eq:complete-second-moment}.  The case $\tau=0$ is immediate.
\end{proof}

The quantification over arbitrary references in
Definition~\ref{def:complete-variance-lift} is essential: a scalar variance
comparison with no reference does not by itself control the complex
cross-terms in Eq.~\eqref{eq:complete-second-moment}.  Proposition
\ref{prop:variance-defect-equivalence} says that this reference-complete
statistical condition has an exact finite-dimensional certificate.  There
are two levels of action:
\[
\tau_A\mathrel{\mathop{\longmapsto}\limits^{\Delta_{\cJ}}}
\Delta_{\cJ}(\tau_A)\in\Lin(\mathsf K_B),
\qquad
Y_B\mathrel{\mathop{\longmapsto}\limits^{\Delta_{\cJ}(\tau_A)}}
\Delta_{\cJ}(\tau_A)(Y_B)\in\mathsf K_B.
\]
Complete positivity concerns the first, outer map.  By the
finite-dimensional Choi theorem~\cite{Choi1975}, it is equivalent to an
outer Kraus representation
\begin{equation}
\Delta_{\cJ}(\tau_A)
=
\sum_\mu G_\mu\tau_A G_\mu^{\ad},
\qquad
G_\mu\in\Lin(\cH_A,\mathsf K_B),
\label{eq:defect-outer-kraus}
\end{equation}
which is the form verified for the qutrit channel in
Section~\ref{sec:certificate}.

\begin{theorem}[BKM domination from complete shadowing]
\normalfont\par\noindent
\label{thm:criterion}
Let $\cN_{A\to B}$ and $\cM_{A\to E}$ be finite-dimensional quantum channels satisfying
\begin{equation}
    \cN(\idop_A)\succ0,
    \qquad
    \cM(\idop_A)\succ0.
    \label{eq:criterion-faithful-identity}
\end{equation}
If $\cM$ provides complete observable shadowing of $\cN$, then
\begin{equation}
    \cM\BKMord\cN.
    \label{eq:criterion-pair-order}
\end{equation}
\end{theorem}

The proof uses the following channel-independent conversion of second-moment
control into domination by the inverse Kubo--Mori weights.

\begin{lemma}[From left-weight domination to BKM domination]
\normalfont\par\noindent
\label{lem:left-weight-to-bkm}
Let $\sigma_B\succ0$ and $\omega_E\succ0$ be faithful states, and let
$\cJ:\mathsf K_B\to\mathsf K_E$ be complex-linear and adjoint preserving.  If
\begin{equation}
    \cJ^{\ad}\circ L_\omega\circ\cJ
    \preceq L_\sigma,
    \label{eq:left-weight-hypothesis}
\end{equation}
then, for every Hermitian $Z_E\in\mathsf K_E$,
\begin{equation}
    g_\omega(Z)
    \ge
    g_\sigma\!\left(\cJ^{\ad}(Z)\right).
    \label{eq:left-weight-bkm-conclusion}
\end{equation}
\end{lemma}

\begin{proof}
Since a complex-linear adjoint-preserving map preserves the Hermitian
subspace, so does its Hilbert--Schmidt adjoint.  In particular,
$\cJ^{\ad}(Z)$ is Hermitian whenever $Z$ is Hermitian.
Adjoint preservation first supplies the corresponding right-weight inequality.  Indeed, Eq.~\eqref{eq:left-weight-hypothesis} applied to
$Y^{\ad}$ gives
\begin{align*}
    \langle Y,R_\sigma(Y)\rangle_{\HS}
    =
    \langle Y^{\ad},L_\sigma(Y^{\ad})\rangle_{\HS}
    \ge
    \langle\cJ(Y^{\ad}),
    L_\omega(\cJ(Y^{\ad}))\rangle_{\HS}
    =
    \langle\cJ(Y),R_\omega(\cJ(Y))\rangle_{\HS}.
\end{align*}
Hence
\begin{equation}
    \cJ^{\ad}\circ R_\omega\circ\cJ
    \preceq R_\sigma.
    \label{eq:right-weight-conclusion}
\end{equation}
For a faithful state $\theta$, put
\begin{equation}
    \mathcal A_{\theta,t}:=(1-t)L_\theta+tR_\theta,
    \qquad 0\le t\le1.
    \label{eq:arithmetic-interpolation}
\end{equation}
In an eigenbasis
$\theta=\sum_i\lambda_i|i\rangle\!\langle i|$,
\begin{equation}
    \mathcal A_{\theta,t}(|i\rangle\!\langle j|)
    =
    \bigl((1-t)\lambda_i+t\lambda_j\bigr)
    |i\rangle\!\langle j|,
    \qquad
    \int_0^1
    \frac{dt}{(1-t)\lambda_i+t\lambda_j}
    =
    \frac{\ln\lambda_i-\ln\lambda_j}
    {\lambda_i-\lambda_j}.
    \label{eq:arithmetic-eigenvalue}
\end{equation}
The quotient is understood as $1/\lambda_i$ when
$\lambda_i=\lambda_j$.  Comparing it with the reciprocal of the logarithmic mean in Eq.~\eqref{eq:kubo-mori-map} gives
\begin{equation}
    \Gamma_\theta^{-1}
    =
    \int_0^1\mathcal A_{\theta,t}^{-1}\,dt.
    \label{eq:bkm-arithmetic}
\end{equation}
Combining Eq.~\eqref{eq:left-weight-hypothesis} and
Eq.~\eqref{eq:right-weight-conclusion} also gives
\begin{equation}
    \cJ^{\ad}\circ\mathcal A_{\omega,t}\circ\cJ
    \preceq\mathcal A_{\sigma,t}
    \qquad(0\le t\le1).
    \label{eq:arithmetic-order}
\end{equation}
For every self-adjoint, strictly positive operator $\mathcal A$ on a
finite-dimensional Hilbert space,
\begin{equation}
    2\operatorname{Re}\langle Y,Z\rangle
    -\langle Y,\mathcal A(Y)\rangle=
    \langle Z,\mathcal A^{-1}(Z)\rangle
    -\bigl\langle
    Y-\mathcal A^{-1}(Z),
    \mathcal A\bigl(Y-\mathcal A^{-1}(Z)\bigr)
    \bigr\rangle.
    \label{eq:inverse-variational}
\end{equation}
The last term is nonnegative and vanishes precisely when
$Y=\mathcal A^{-1}(Z)$.  Consequently,
\[
\langle Z,\mathcal A^{-1}(Z)\rangle
=
\sup_Y
\left\{
2\operatorname{Re}\langle Y,Z\rangle
-\langle Y,\mathcal A(Y)\rangle
\right\}.
\]
Fix $t\in[0,1]$.  Applying
Eq.~\eqref{eq:inverse-variational} first on $\mathsf K_E$ and then on
$\mathsf K_B$ gives
\begin{align}
    \langle Z,\mathcal A_{\omega,t}^{-1}(Z)\rangle
    &=
    \sup_{W\in\mathsf K_E}
    \left\{
    2\operatorname{Re}\langle W,Z\rangle
    -\langle W,\mathcal A_{\omega,t}(W)\rangle
    \right\}\notag\\
    &\ge
    \sup_{Y\in\mathsf K_B}
    \left\{
    2\operatorname{Re}\langle\cJ(Y),Z\rangle
    -\langle\cJ(Y),
    \mathcal A_{\omega,t}(\cJ(Y))\rangle
    \right\}\notag\\
    &\ge
    \sup_{Y\in\mathsf K_B}
    \left\{
    2\operatorname{Re}
    \langle Y,\cJ^{\ad}(Z)\rangle
    -\langle Y,\mathcal A_{\sigma,t}(Y)\rangle
    \right\}\notag\\
    &=
    \left\langle
    \cJ^{\ad}(Z),
    \mathcal A_{\sigma,t}^{-1}
    \bigl(\cJ^{\ad}(Z)\bigr)
    \right\rangle.
    \label{eq:inverse-weight-transfer}
\end{align}
The first inequality restricts the environment variational domain to
$W=\cJ(Y)$.  The second uses
Eq.~\eqref{eq:arithmetic-order} together with
$\langle\cJ(Y),Z\rangle_{\HS}
=\langle Y,\cJ^{\ad}(Z)\rangle_{\HS}$.

Integrating Eq.~\eqref{eq:inverse-weight-transfer} over $t$ and using
Eq.~\eqref{eq:bkm-arithmetic} yields
\begin{align}
    g_\omega(Z)
    =
    \int_0^1
    \langle Z,\mathcal A_{\omega,t}^{-1}(Z)\rangle_{\HS}\,dt
    \ge
    \int_0^1
    \left\langle
    \cJ^{\ad}(Z),
    \mathcal A_{\sigma,t}^{-1}
    \bigl(\cJ^{\ad}(Z)\bigr)
    \right\rangle_{\HS}\,dt
    =
    g_\sigma\!\left(\cJ^{\ad}(Z)\right),
\end{align}
proving Eq.~\eqref{eq:left-weight-bkm-conclusion}.
\end{proof}

\begin{proof}[Proof of Theorem~\ref{thm:criterion}]
Fix a finite reference $R$, a faithful state $\rho_{RA}\succ0$, and a
traceless Hermitian $X_{RA}$.  Set
\[
\sigma_{RB}:=\cN_R(\rho),\qquad
\omega_{RE}:=\cM_R(\rho),\qquad
\cJ_R:=\id_R\otimes\cJ.
\]
Since $\rho_{RA}\succeq a\,\idop_R\otimes\idop_A$ for some $a>0$,
Eq.~\eqref{eq:criterion-faithful-identity} gives
\begin{equation}
    \sigma_{RB}
    \succeq a\,\idop_R\otimes\cN(\idop_A)\succ0,
    \qquad
    \omega_{RE}
    \succeq a\,\idop_R\otimes\cM(\idop_A)\succ0.
    \label{eq:criterion-faithful-outputs}
\end{equation}
The exact reconstruction identity amplifies to
\begin{equation}
    \cJ_R^{\ad}\circ\cM_R
    =
    \id_R\otimes(\cJ^{\ad}\circ\cM)
    =
    \cN_R.
    \label{eq:criterion-amplified-lift}
\end{equation}
By Proposition~\ref{prop:variance-defect-equivalence}, complete variance
domination is equivalent to Eq.~\eqref{eq:complete-second-moment}.  Taking
$\tau=\rho$ in that inequality gives
\[
\cJ_R^{\ad}\circ L_\omega\circ\cJ_R
\preceq L_\sigma.
\]
The map $\cJ_R$ is adjoint preserving, so Lemma~\ref{lem:left-weight-to-bkm}, applied with $Z=\cM_R(X)$, gives
\begin{align}
    g_\omega\!\left(\cM_R(X)\right)
    \ge
    g_\sigma\!\left(\cJ_R^{\ad}(\cM_R(X))\right) =
    g_\sigma\!\left(\cN_R(X)\right),
\end{align}
where the equality is Eq.~\eqref{eq:criterion-amplified-lift}.
Since $R$, $\rho$, and $X$ were arbitrary, $\cM\BKMord\cN$.
\end{proof}

\begin{lemma}[Reduction to output supports]
\normalfont\par\noindent
\label{lem:output-support-reduction}
Let $\cN_{A\to B}$ and $\cM_{A\to E}$ be quantum channels, and suppose
that $\cM$ provides complete observable shadowing of $\cN$.
Set $B_0:=\supp\cN(\idop_A)$ and $E_0:=\supp\cM(\idop_A)$.
Let $U_B:B_0\to B$ and $U_E:E_0\to E$ be the inclusion isometries, and define
\begin{equation}
 \cN_0(X):=U_B^\dagger\cN(X)U_B,\qquad
 \cM_0(X):=U_E^\dagger\cM(X)U_E.
 \label{eq:support-reduced-channels}
\end{equation}
These are channels with full-rank identity outputs on $B_0$ and $E_0$.
Moreover, $\cM_0$ provides complete observable shadowing of $\cN_0$.
Consequently, complete observable shadowing implies
$\cM\REord\cN$ without any full-rank assumption on the original output spaces.
\end{lemma}

\begin{proof}
We first show that restricting the output spaces to $B_0$ and $E_0$ discards no part of the channel outputs. For every $X\succeq0$,
positivity gives $0\preceq\cN(X)\preceq\|X\|_\infty\cN(\idop_A)$, so $\cN(X)$ is supported on $B_0$. The analogous statement holds for $\cM(X)$ and $E_0$. Let $P_B:=U_BU_B^\dagger$ and $P_E:=U_EU_E^\dagger$ be the orthogonal projections onto $B_0$ and $E_0$, respectively.
The preceding argument shows that
\[
P_B\cN(X)P_B= \cN(X),
\qquad
P_E\cM(X)P_E=\cM(X)
\]
for every positive semidefinite $X$. Since positive semidefinite operators span the full operator space, these identities extend
by linearity to every $X\in\Lin(\cH_A)$. Therefore,
\[
U_B\cN_0(X)U_B^\dagger
=P_B\cN(X)P_B
=\cN(X),
\]
and similarly
$U_E\cM_0(X)U_E^\dagger=\cM(X)$. The reduced maps are also completely positive and trace preserving. Moreover, $\cN_0(\idop_A)$ and $\cM_0(\idop_A)$ are strictly positive
on $B_0$ and $E_0$, respectively, by the definitions of these supports.

Let $\cJ$ be a shadow map for the original channel pair. We obtain a
map between the reduced observable spaces by embedding a receiver
observable into $B$, applying $\cJ$, and compressing the result to $E_0$:
\begin{equation}
\cJ_0(Y):=U_E^\dagger\cJ(U_BYU_B^\dagger)U_E.
\label{eq:support-reduced-shadow}
\end{equation}
This map is complex-linear and adjoint preserving. Its
Hilbert--Schmidt adjoint, together with the output-support identities
above, gives
\[
\begin{aligned}
\cJ_0^\dagger\!\left(\cM_0(X)\right)
&=
U_B^\dagger
\cJ^\dagger\!\left(U_E\cM_0(X)U_E^\dagger\right)U_B\\
&=
U_B^\dagger\cJ^\dagger\!\left(\cM(X)\right)U_B
=
U_B^\dagger\cN(X)U_B
=
\cN_0(X).
\end{aligned}
\]
Thus, the reduced map satisfies the required unbiasedness condition.

To verify the variance comparison, the key observation is that
compressing an observable to the support of a state preserves its mean
and can only decrease its variance. More precisely, let $U$ be an
isometry, let $\theta_0$ be a state on its domain, and let $T$ be a
Hermitian observable on its codomain. The means satisfy
$\Tr(U\theta_0U^\dagger T)=\Tr(\theta_0U^\dagger TU)$, while
\[
U^\dagger T^2U-(U^\dagger TU)^2
=
U^\dagger T(I-UU^\dagger)TU\succeq0.
\]
The positivity follows because $I-UU^\dagger$ is an orthogonal
projection. Subtracting the common squared mean therefore gives
\begin{equation}
\operatorname{Var}_{U\theta_0U^\dagger}(T)
-
\operatorname{Var}_{\theta_0}(U^\dagger TU)
=
\Tr\!\left[\theta_0U^\dagger T(I-UU^\dagger)TU\right]\geq0.
\label{eq:compression-second-moment}
\end{equation}

Now fix an arbitrary finite reference $R$, an input state $\rho_{RA}$,
and a Hermitian observable $Y$ on $RB_0$. Set
$U_{B,R}:=\idop_R\otimes U_B$ and
$U_{E,R}:=\idop_R\otimes U_E$, and define
\begin{equation}
\widehat Y:=U_{B,R}YU_{B,R}^\dagger,
\qquad
T:=(\id_R\otimes\cJ)(\widehat Y).
\label{eq:compressed-observable}
\end{equation}
Then $T$ is Hermitian and
\[
U_{E,R}^\dagger TU_{E,R}
=
(\id_R\otimes\cJ_0)(Y).
\]
Apply Eq.~\eqref{eq:compression-second-moment} to the environment
output, with $U=U_{E,R}$ and
$\theta_0=(\id_R\otimes\cM_0)(\rho)$. Combining the resulting
inequality with the original complete variance comparison yields
\[
\begin{aligned}
\operatorname{Var}_{(\id_R\otimes\cM_0)(\rho)}
\!\left((\id_R\otimes\cJ_0)(Y)\right)
&\leq
\operatorname{Var}_{(\id_R\otimes\cM)(\rho)}(T)\\
&\leq
\operatorname{Var}_{(\id_R\otimes\cN)(\rho)}(\widehat Y)\\
&=
\operatorname{Var}_{(\id_R\otimes\cN_0)(\rho)}(Y).
\end{aligned}
\]
The last equality holds because both the receiver state and observable
are embedded by the same isometry $U_{B,R}$, which preserves their
first and second moments. Since $R$, $\rho$, and $Y$ were arbitrary,
$\cJ_0$ witnesses complete observable shadowing of $\cN_0$ by $\cM_0$.

The reduced channels now satisfy the full-rank hypotheses of
Theorem~\ref{thm:criterion} and
Lemma~\ref{lem:complete-bkm-to-cre}, so
$\cM_0\BKMord\cN_0$ and hence $\cM_0\REord\cN_0$.
At every reference extension, the original outputs are isometric
embeddings of the reduced outputs. Invariance of relative entropy
under these embeddings therefore gives $\cM\REord\cN$.
\end{proof}

Combining the preceding results gives the capacity consequence without a restriction on output supports.

\begin{corollary}[Zero capacities from complete observable shadowing]
\normalfont\par\noindent
\label{cor:criterion-consequence}
Let $\cN_{A\to B}$ be a finite-dimensional quantum channel with
complementary channel $\cN^c_{A\to E}$.
If $\cN^c$ provides complete observable shadowing of $\cN$, then
\begin{equation}
    \cN^c\REord\cN.
    \label{eq:criterion-re-order}
\end{equation}
Consequently, for every $n\ge1$,
\begin{equation}
    P^{(1)}(\cN^{\otimes n})
    =
    Q^{(1)}(\cN^{\otimes n})
    =
    0,
    \qquad
    P(\cN)=Q(\cN)=0.
    \label{eq:criterion-capacity-consequence}
\end{equation}
\end{corollary}

\begin{proof}
Lemma~\ref{lem:output-support-reduction} gives $\cN^c\REord\cN$.
The capacity conclusions follow from
Corollary~\ref{cor:anti-less-noisy-capacity}.
\end{proof}

\begin{remark}[Operational meaning]
\normalfont\par\noindent
For a fixed input state and a fixed Hermitian receiver observable $Y_B$, the measurement of the Hermitian operator $\cJ(Y_B)$ on independent outputs of $\cN^{c}$ gives an unbiased estimator of $\Tr[\cN(\rho)Y_B]$.  Its sample-mean variance is no larger than that obtained by measuring $Y_B$ on independent outputs of $\cN$.

This statement is observable by observable.  Since $\cJ$ need not be
positive, the definition does not require a physical simulation of the
receiver state or of arbitrary receiver POVM outcome distributions.  The reference system in
Eq.~\eqref{eq:complete-variance-domination} is likewise a stability test:
the comparison remains valid for joint, possibly nonproduct observables on a
retained reference and the respective channel output. It does not grant the
physical environment access to that reference.
\end{remark}

\begin{remark}[Unbiased reconstruction alone is insufficient]
\normalfont\par\noindent
\label{rem:unbiased-not-enough}
For $0<\gamma<1/2$, consider the qubit amplitude-damping channel
\begin{equation}
 \mathcal A_\gamma
 \begin{pmatrix}x_{00}&x_{01}\\x_{10}&x_{11}\end{pmatrix}
 =\begin{pmatrix}
 x_{00}+\gamma x_{11}&\sqrt{1-\gamma}\,x_{01}\\
 \sqrt{1-\gamma}\,x_{10}&(1-\gamma)x_{11}
 \end{pmatrix}.
 \label{eq:amplitude-damping-counterexample}
\end{equation}
The standard isometry gives the complement
$\mathcal A_\gamma^c=\mathcal A_{1-\gamma}$, which is invertible as a
linear map for $\gamma>0$. Hence
$\cJ^\dagger:=\mathcal A_\gamma\circ(\mathcal A_\gamma^c)^{-1}$ is a
complex-linear, adjoint-preserving reconstruction satisfying
$\cJ^\dagger\mathcal A_\gamma^c=\mathcal A_\gamma$.
It reconstructs every observable expectation, also under reference
extensions. However, direct evaluation at the maximally mixed input gives
\begin{equation}
 I_{\mathrm c}(\idop_2/2,\mathcal A_\gamma)
 =h_2((1-\gamma)/2)-h_2(\gamma/2)>0,
 \label{eq:amplitude-damping-positive-ci}
\end{equation}
because $h_2$ is strictly increasing on $(0,1/2)$. The quantum coding theorem
therefore gives $Q(\mathcal A_\gamma)>0$. The reconstruction cannot satisfy
the complete variance comparison. This example isolates the role of variance
control in the shadowing criterion.
\end{remark}

\section{Zero private and quantum capacities of the qutrit channel}
\label{sec:certificate}

For the qutrit channel, both identity outputs have full rank, so we can apply
Theorem~\ref{thm:criterion} directly. To establish
Eq.~\eqref{eq:main-lambda-zero}, we construct an adjoint-preserving map with
the required reconstruction identity and verify its complete variance bound.
By Proposition~\ref{prop:variance-defect-equivalence}, the latter reduces to
an exact Kraus factorization of the outer defect.

\subsection{Channel structure and a complementary channel}

We work with qutrit systems
$\cH_A\simeq\cH_B\simeq\cH_3:=\mathbb C^3$.  The channel under
consideration is
\begin{equation}
\Lambda_{A\to B}(X_A)
=
\frac12X_A+\frac14
\bigl(\Tr(X_A)\idop_B-X_A^{\mathsf T}\bigr).
\label{eq:channel}
\end{equation}
It is the $x=\tfrac12$ member of the noisy Werner--Holevo
family~\cite{WernerHolevo2002,RoofehKarimipour2024}.  We suppress input--output
labels when no ambiguity can arise. Every $X_A\in\Lin(\cH_A)$ admits the orthogonal decomposition
\begin{equation}
X_A=\frac{\Tr X_A}{3}\idop_A+S_A+T_A,
\qquad
S_A^{\mathsf T}=S_A,\quad \Tr S_A=0,\quad
T_A^{\mathsf T}=-T_A.
\label{eq:operator-decomp}
\end{equation}
Substitution into Eq.~\eqref{eq:channel} gives
\begin{equation}
\Lambda(X_A)
=
\frac{\Tr X_A}{3}\idop_B+\frac14S_A+\frac34T_A.
\label{eq:sector-action}
\end{equation}
Thus $\Lambda$ acts diagonally on the scalar, symmetric-traceless, and
skew-symmetric operator subspaces, whose dimensions are $1$, $5$, and
$3$, respectively.  This decomposition provides a useful structural
picture of the channel and helps motivate the shadow map constructed
below.

We next fix a complementary channel adapted to this representation.
Let $\epsilon_{\mu jk}$ be the Levi--Civita tensor, with indices in
$\{0,1,2\}$.  Relative to the fixed bases, define
$A_\mu\in\Lin(\cH_3)$ by
\begin{equation}
(A_\mu)_{jk}:=\epsilon_{\mu jk},
\qquad \mu=0,1,2.
\label{eq:A-mu}
\end{equation}
The contraction
\begin{equation}
\sum_{\mu=0}^2
\epsilon_{\mu jk}\epsilon_{\mu\ell m}
=
\delta_{j\ell}\delta_{km}-\delta_{jm}\delta_{k\ell}
\label{eq:levi-civita-contraction}
\end{equation}
fixes the transpose and sign conventions in the identities below.
These real skew-symmetric operators satisfy
\begin{equation}
\sum_{\mu=0}^2A_\mu X A_\mu^{\ad}
=\Tr(X)\idop_3-X^{\mathsf T},
\qquad
\sum_{\mu=0}^2A_\mu^{\ad}A_\mu=2\idop_3.
\label{eq:epsilon-contractions}
\end{equation}
Hence
\begin{equation}
K_0=\frac{\idop_3}{\sqrt2},
\qquad
K_{\mu+1}=\frac{A_\mu}{2},
\quad \mu=0,1,2,
\label{eq:kraus}
\end{equation}
is a Kraus representation of $\Lambda$, where the fixed identification $\cH_A\simeq\cH_B\simeq\cH_3$ is understood.  This is the $x=\tfrac12$, $d=3$ specialization of a previously established Kraus and complementary-channel construction~\cite{RoofehKarimipour2024}.

Let $\cH_E\simeq\mathbb C^4$, with computational basis $\{|a\rangle_E\}_{a=0}^3$.  Fix the associated complementary channel 
\begin{equation}
\Lambda^c_{A\to E}:\Lin(\cH_A)\to\Lin(\cH_E),\qquad
{}_E\langle a|\Lambda^c_{A\to E}(X_A)|b\rangle_E
:=\Tr(K_aX_AK_b^{\ad}),
\qquad 0\le a,b\le3.
\label{eq:complement}
\end{equation}
For $Y\in\Lin(\cH_3)$, set
\begin{equation}
\boldsymbol a(Y):=
\begin{pmatrix}
    Y_{12}-Y_{21}\\
    Y_{20}-Y_{02}\\
    Y_{01}-Y_{10}
\end{pmatrix}.
\label{eq:a-vector}
\end{equation}
A direct contraction gives
\begin{equation}
\Lambda^c(X_A)=
\begin{pmatrix}
    \frac12\Tr X_A &
    \frac1{2\sqrt2}\boldsymbol a(X_A)^{\mathsf T}\\[1mm]
    -\frac1{2\sqrt2}\boldsymbol a(X_A) &
    \frac14\bigl(\Tr(X_A)\idop_3-X_A^{\mathsf T}\bigr)
\end{pmatrix}.
\label{eq:complement-block}
\end{equation}
In particular,
\begin{equation}
\Lambda(\idop_A)=\idop_B,\qquad
\Lambda^c(\idop_A)=
\operatorname{diag}\!\left(\frac32,\frac12,\frac12,\frac12\right)\succ0.
\label{eq:faithful-identity-output}
\end{equation}
The strict positivity in Eq.~\eqref{eq:faithful-identity-output} guarantees that faithful inputs produce faithful receiver and environment outputs.

Writing
$Z=\left(\begin{smallmatrix}\alpha&u^{\mathsf T}\\v&C\end{smallmatrix}\right)$, the Hilbert--Schmidt adjoint is
\begin{equation}
(\Lambda^c)^{\ad}(Z)
=\frac\alpha2\idop_3
+\frac1{2\sqrt2}\sum_{\mu=0}^2(u_\mu-v_\mu)A_\mu
+\frac14\bigl[(\Tr C)\idop_3-C^{\mathsf T}\bigr].
\label{eq:complement-adjoint}
\end{equation}

\subsection{A shadow map for the qutrit channel}

Define the complex-linear map
$\cJ_{B\to E}:\Lin(\cH_B)\to\Lin(\cH_E)$ by
\begin{equation}
\cJ_{B\to E}(Y_B)=\frac13
\begin{pmatrix}
    \Tr Y_B & \sqrt2\,\boldsymbol a(Y_B)^{\mathsf T}\\[1mm]
    -\sqrt2\,\boldsymbol a(Y_B)&
    2\Tr(Y_B)\idop_B-Y_B-2Y_B^{\mathsf T}
\end{pmatrix}.
\label{eq:J-lift}
\end{equation}
We suppress the system label on $\cJ$ below.

\begin{lemma}[Exact observable reconstruction]
\normalfont\par\noindent
\label{lem:lift}
The map $\cJ$ in Eq.~\eqref{eq:J-lift} is adjoint preserving and satisfies
\begin{equation}
    (\Lambda^c)^{\ad}\circ\cJ=\Lambda^{\ad}=\Lambda,
    \qquad
    \cJ(Y^{\ad})=\cJ(Y)^{\ad}.
    \label{eq:lift-identities}
\end{equation}
Equivalently,
\begin{equation}
    \cJ^{\ad}\circ\Lambda^c=\Lambda.
    \label{eq:adjoint-lift}
\end{equation}
The map $\cJ$ is not positive.
\end{lemma}

\begin{proof}
The identity
\begin{equation}
    \sum_{\mu=0}^2a_\mu(Y)A_\mu=Y-Y^{\mathsf T}
    \label{eq:aA}
\end{equation}
follows from the Levi--Civita contraction.  Put $t=\Tr Y$ and
$C_Y=\frac13(2t\idop_3-Y-2Y^{\mathsf T})$, so $\Tr C_Y=t$.  Substitution into
Eq.~\eqref{eq:complement-adjoint} gives
\begin{align}
    \bigl((\Lambda^c)^{\ad}\circ\cJ\bigr)(Y)
    &=\frac t6\idop_3+\frac13(Y-Y^{\mathsf T})
    +\frac1{12}(t\idop_3+Y^{\mathsf T}+2Y)\notag\\
    &=\frac12Y+\frac14(t\idop_3-Y^{\mathsf T})
    =\Lambda(Y).
\end{align}
Adjoint preservation follows from
$\boldsymbol a(Y^{\ad})=-\overline{\boldsymbol a(Y)}$.  Finally,
\[
\cJ(E_{00})=\frac13\operatorname{diag}(1,-1,2,2)
\]
is not positive.
\end{proof}

For every input state $\rho$ and Hermitian receiver observable
$Y=Y^{\ad}$,
\begin{equation}
\Tr\!\left[Y\Lambda(\rho)\right]
=
\Tr\!\left[\cJ(Y)\Lambda^c(\rho)\right].
\label{eq:signed-reconstruction}
\end{equation}
Because $\cJ$ is adjoint preserving, $\cJ(Y)$ is a legitimate Hermitian
environment observable whenever $Y$ is Hermitian. Thus the environment can
reproduce the expectation of every receiver observable by measuring the
corresponding $\cJ(Y)$. The map is not positive, so even a receiver effect
can be mapped to a Hermitian operator with negative eigenvalues. In this
sense the estimator is \emph{signed}; reconstructing its expectation does not
by itself simulate the receiver's measurement outcomes.

It remains to verify that this observable-by-observable reconstruction is
variance contracting in the presence of arbitrary references.  By
Proposition~\ref{prop:variance-defect-equivalence}, it suffices to show that
the outer defect of the shadow map is completely positive.
Let $B'$ be a copy of $B$, and make the Hilbert--Schmidt identification
$\Lin(\cH_B)\simeq\cH_B\otimes\cH_{B'}$ explicit through
$E^B_{jk}\leftrightarrow|j\rangle_B|k\rangle_{B'}$.  For
$\tau_A\in\Lin(\cH_A)$, define
\begin{equation}
\Delta(\tau_A)
:= L_{\Lambda(\tau_A)}
-\cJ^{\ad}\circ L_{\Lambda^c(\tau_A)}\circ\cJ
\in\Lin\bigl(\Lin(\cH_B)\bigr).
\label{eq:defect}
\end{equation}
For $\mu=0,1,2$, define
$G_\mu:\cH_A\to\cH_B\otimes\cH_{B'}$ by
\begin{equation}
G_\mu|i\rangle_A
=
|i\rangle_B|\mu\rangle_{B'}
+\frac12|\mu\rangle_B|i\rangle_{B'}
-\frac12\delta_{\mu i}
\sum_{j=0}^2|j\rangle_B|j\rangle_{B'}.
\label{eq:G-action}
\end{equation}

\begin{proposition}[Exact CP defect]
\normalfont\par\noindent
\label{prop:cp-defect}
For every $\tau_A\in\Lin(\cH_A)$,
\begin{equation}
    \Delta(\tau_A)
    =\frac29\sum_{\mu=0}^2G_\mu \tau_A G_\mu^{\ad}.
    \label{eq:cp-defect}
\end{equation}
Hence
$\Delta:\Lin(\cH_A)\to\Lin(\Lin(\cH_B))$ is completely positive.
\end{proposition}

\begin{proof}
Both sides of Eq.~\eqref{eq:cp-defect} are linear in $\tau$.  Expanding on matrix
units gives the same coefficients. The complete entrywise calculation is
recorded in Section~\ref{app:defect-expansion}.  Since the right-hand side is a
Kraus representation, complete positivity follows.
\end{proof}

Together with Lemma~\ref{lem:lift}, Proposition~\ref{prop:cp-defect} and
Proposition~\ref{prop:variance-defect-equivalence} show that $\cJ$ is a
shadow map witnessing complete observable shadowing of $\Lambda$ by
$\Lambda^c$.
Thus the exact factorization certifies not only the following scalar variance
comparison, but its stability under every reference system and every joint
Hermitian observable.

The factorization has a more transparent quadratic form.  Define
\begin{equation}
r_\mu(Y)
:= Y|\mu\rangle+\frac12Y^{\mathsf T}|\mu\rangle
-\frac12\Tr(Y)|\mu\rangle.
\label{eq:rmu}
\end{equation}
Then $r_\mu(Y)=G_\mu^{\ad}|Y\rangle_{\HS}$, so for every $\tau\succeq0$,
\begin{equation}
\langle Y,\Delta(\tau)Y\rangle_{\HS}
=
\frac29\sum_{\mu=0}^2
\langle r_\mu(Y),\tau\,r_\mu(Y)\rangle
\ge0.
\label{eq:defect-sos}
\end{equation}
For a state $\rho$ and a Hermitian observable $Y$, Eq.~\eqref{eq:signed-reconstruction} and
Eq.~\eqref{eq:defect-sos} give the exact variance identity
\begin{equation}
\operatorname{Var}_{\Lambda(\rho)}(Y)
-
\operatorname{Var}_{\Lambda^c(\rho)}\!\left(\cJ(Y)\right)
=
\frac29\sum_{\mu=0}^2
\left\|\rho^{1/2}r_\mu(Y)\right\|_2^2.
\label{eq:qutrit-variance-sos}
\end{equation}
The three residual vectors $r_\mu(Y)$ therefore quantify exactly how much
more variable the receiver observable is than its signed environment
reconstruction.  If $\rho\succ0$, equality in
Eq.~\eqref{eq:qutrit-variance-sos} holds only for scalar observables: the
relations $r_\mu(Y)=0$ for all $\mu$ are equivalent to
$2Y+Y^{\mathsf T}=\Tr(Y)\idop_B$, which forces
$Y=\Tr(Y)\idop_B/3$.

As a concrete example, take $\rho=\idop_A/3$ and
$Y=E_{00}$.  In this case
$\Lambda(\rho)=\idop_B/3$ and
$\Lambda^c(\rho)=\operatorname{diag}(1/2,1/6,1/6,1/6)$, while
\begin{equation}
\cJ(E_{00})
=
\frac13\operatorname{diag}(1,-1,2,2),
\qquad
\operatorname{Var}_{\Lambda(\rho)}(E_{00})=\frac29,
\qquad
\operatorname{Var}_{\Lambda^c(\rho)}\!\left(\cJ(E_{00})\right)
=\frac19.
\label{eq:qutrit-variance-example}
\end{equation}
The negative eigenvalue makes clear why $\cJ(E_{00})$ is not an effect.
It is instead a signed estimator whose expectation agrees with the
receiver probability and whose variance is smaller.

\subsection{Complete information domination and zero capacities}

We now prove Eq.~\eqref{eq:main-lambda-zero}
without invoking either PPT structure or antidegradability.

\begin{theorem}[Qutrit zero capacities]
\normalfont\par\noindent
\label{thm:qutrit-zero-capacity}
For every finite-dimensional reference system $R$, every $n\ge1$, and
all states $\rho_{RA^n},\sigma_{RA^n}$ satisfying
$\supp\rho_{RA^n}\subseteq\supp\sigma_{RA^n}$,
\begin{align}
    D\!\left((\id_R\otimes(\Lambda^c)^{\otimes n})(\rho_{RA^n})
    \middle\Vert
    (\id_R\otimes(\Lambda^c)^{\otimes n})(\sigma_{RA^n})\right)
    \ge
    D\!\left((\id_R\otimes\Lambda^{\otimes n})(\rho_{RA^n})
    \middle\Vert
    (\id_R\otimes\Lambda^{\otimes n})(\sigma_{RA^n})\right).
    \label{eq:qutrit-all-block-order}
\end{align}
Moreover,
\begin{equation}
    P^{(1)}(\Lambda^{\otimes n})
    =Q^{(1)}(\Lambda^{\otimes n})=0
    \qquad(n\ge1),
    \label{eq:qutrit-all-block-zero}
\end{equation}
and consequently
\begin{equation}
    P(\Lambda)=Q(\Lambda)=0.
    \label{eq:qutrit-capacity-zero}
\end{equation}
\end{theorem}

\begin{proof}
Eq.~\eqref{eq:faithful-identity-output} verifies the full-rank
identity-output condition in Theorem~\ref{thm:criterion}.  By
Lemma~\ref{lem:lift},~Propositions~\ref{prop:cp-defect} and~\ref{prop:variance-defect-equivalence}, the map $\cJ$
is a shadow map witnessing complete observable shadowing of $\Lambda$ by
$\Lambda^c$. Apply Theorem~\ref{thm:criterion} with
$\cM=\Lambda^c$ and $\cN=\Lambda$ to obtain
$\Lambda^c\BKMord\Lambda$.  The BKM-to-relative-entropy
comparison in Lemma~\ref{lem:complete-bkm-to-cre} gives
$\Lambda^c\REord\Lambda$.  The tensor-power statement in
Eq.~\eqref{eq:qutrit-all-block-order} follows by applying
Eq.~\eqref{eq:tensorized-order}.  The
zero-information and zero-capacity conclusions follow from
Corollary~\ref{cor:anti-less-noisy-capacity}.
\end{proof}

\section{Separation from PPT and antidegradable channels}
\label{sec:separation}

We show that the normalized Choi state has a negative partial transpose
(NPT) and rule out antidegradability directly from the outputs of the
channel and its complement on three input states.

\subsection{NPT Choi state}

Let $A'$ be a copy of the channel input system $A$. For any pair $X,Y$ of qutrit systems equipped with the fixed computational bases, write
\begin{equation}
|\Phi\rangle_{XY}
:=\frac1{\sqrt3}\sum_{i=0}^2|i\rangle_X|i\rangle_Y,\qquad
\Phi_{XY}:=|\Phi\rangle\!\langle\Phi|_{XY},
\label{eq:phi3}
\end{equation}
and let
$F_{XY}:=\sum_{i,j=0}^2\ketbra{i}{j}_X\otimes\ketbra{j}{i}_Y$
be the swap.  On $\cH_A\otimes\cH_B$, define
\begin{equation}
P^-_{AB}:=\frac{\idop_{AB}-F_{AB}}{2}
\label{eq:pminus}
\end{equation}
as the projector onto $\wedge^2\cH_3$.

A direct application of Eq.~\eqref{eq:channel} gives
\begin{equation}
\omega^\Lambda_{AB}
:=(\id_A\otimes\Lambda_{A'\to B})(\Phi_{AA'})
=\frac12\Phi_{AB}+\frac16P^-_{AB},
\label{eq:choi-decomposition}
\end{equation}
because
\[
(\id_A\otimes\Lambda)(\Phi)
=\frac12\Phi+\frac1{12}(\idop-F).
\]
Using $\Phi^{T_A}=F/3$ and $F^{T_A}=3\Phi$ then gives
\begin{equation}
(\omega^\Lambda_{AB})^{T_A}
=\frac1{12}\bigl(\idop_{AB}+2F_{AB}-3\Phi_{AB}\bigr),
\qquad
\spec\bigl((\omega^\Lambda_{AB})^{T_A}\bigr)
=
\left\{
\left(\frac14\right)^{\!\times5},
0,
\left(-\frac1{12}\right)^{\!\times3}
\right\}.
\label{eq:choi-pt}
\end{equation}
In particular, the Choi state is NPT.  These formulas agree with the
expressions for the noisy Werner--Holevo family in
Eqs.~(32)--(34) and (42) of Ref.~\cite{RoofehKarimipour2025TwoParameter}.
The Choi state is also one-copy distillable: local projections on one copy
can produce an entangled two-qubit state with nonzero probability. This does
not assert that an exact singlet can be extracted from one copy.

\begin{proposition}[One-copy distillability]
\normalfont\par\noindent
\label{prop:choi}
The normalized Choi state $\omega^\Lambda_{AB}$ is one-copy distillable.
\end{proposition}

\begin{proof}
The antisymmetric vector
\[
|\psi^-_{01}\rangle_{AB}
:=\frac{|0\rangle_A|1\rangle_B-|1\rangle_A|0\rangle_B}{\sqrt2}
\]
has Schmidt rank two and satisfies
\[
\langle\psi^-_{01}|_{AB}
(\omega^\Lambda_{AB})^{T_A}
|\psi^-_{01}\rangle_{AB}=-\frac1{12}<0.
\]
The Schmidt-rank-two criterion therefore proves one-copy
distillability~\cite{HorodeckiDistillation1998}.
\end{proof}

\begin{remark}[One-way versus two-way distillation]
\normalfont\par\noindent
Let $D_{\to}$ denote distillable entanglement
under one-way classical communication from $A$ to $B$.  Although
$\omega^\Lambda_{AB}$ is one-copy distillable, its one-way distillable entanglement vanishes. Indeed, one copy of $\omega^\Lambda_{AB}$ can be generated from one use of $\Lambda$ by preparing $\Phi_{AA'}$ and transmitting $A'$ through the channel.  Consequently, any one-way distillation protocol for
$(\omega^\Lambda_{AB})^{\otimes n}$ induces an entanglement-generation
protocol for $\Lambda^{\otimes n}$ assisted by forward classical
communication.  Hence
\begin{equation}
    D_{\to}(\omega^\Lambda_{AB})
    \le Q_{\to}(\Lambda),
\end{equation}
where $Q_{\to}$ denotes the forward-classical-assisted quantum capacity.
Forward classical communication does not increase quantum
capacity~\cite{BarnumKnillNielsen2000}, and therefore
\begin{equation}
    D_{\to}(\omega^\Lambda_{AB})
    \le Q_{\to}(\Lambda)
    =Q(\Lambda)
    =0.
\end{equation}
\end{remark}

\subsection{Failure of antidegradability}
\label{subsec:non-antidegradable}

We rule out antidegradability using only the three computational-basis
inputs. The argument relies on positivity and trace preservation, without
requiring a symmetric-extension criterion.

\begin{proposition}[Failure of antidegradability]
\normalfont\par\noindent
\label{prop:non-antidegradable}
The channel $\Lambda$ is not antidegradable.
\end{proposition}

\begin{proof}
Suppose, for contradiction, that a quantum channel
$\cD:\Lin(\cH_E)\to\Lin(\cH_B)$ satisfies
$\cD\circ\Lambda^c=\Lambda$, where $\Lambda^c$ is the complement in
Eq.~\eqref{eq:complement-block}.
Write $P_i:=\ketbra ii$ for the qutrit computational-basis projectors,
with the fixed input--output identification understood, and
$Q_a:=\ketbra aa_E$ for the environment projectors.
For $i=0,1,2$, Eqs.~\eqref{eq:channel} and
\eqref{eq:complement-block} give
\begin{equation}
\Lambda(P_i)=\frac14(\idop_B+P_i),
\qquad
\Lambda^c(P_i)=\frac12Q_0
+\frac14\sum_{\substack{j=0\\j\neq i}}^2Q_{j+1}.
\label{eq:nonad-basis-outputs}
\end{equation}
Set
\[
\sigma:=\cD(Q_0),
\qquad
\tau_i:=\cD(Q_{i+1}),\qquad i=0,1,2.
\]
These are density operators because $\cD$ is positive and trace
preserving. Applying $\cD$ to the complementary outputs in
Eq.~\eqref{eq:nonad-basis-outputs} and multiplying by four yields
\begin{equation}
2\sigma+\sum_{j\neq i}\tau_j=\idop_B+P_i,
\qquad i=0,1,2.
\label{eq:nonad-reconstruction-constraints}
\end{equation}
Summing these three identities and dividing by two gives
\[
3\sigma+\sum_{j=0}^2\tau_j=2\idop_B.
\]
Subtracting Eq.~\eqref{eq:nonad-reconstruction-constraints} for a fixed
$i$ therefore gives
\begin{equation}
\sigma+\tau_i=\idop_B-P_i.
\label{eq:nonad-support-constraint}
\end{equation}
Taking the expectation in $\ket i$ and using positivity, we obtain
\[
0\leq\langle i|\sigma|i\rangle
\leq\langle i|(\sigma+\tau_i)|i\rangle
=\langle i|(\idop_B-P_i)|i\rangle
=0.
\]
Thus, all three diagonal entries of $\sigma$ vanish, so
$\Tr\sigma=0$. This contradicts $\Tr\sigma=\Tr Q_0=1$.
Hence, no antidegrading channel exists.
\end{proof}

\section{Entrywise verification of the CP defect}
\label{app:defect-expansion}

For completeness, we verify the identity in
Eq.~\eqref{eq:cp-defect} directly.  In the
identification
$E^B_{jk}\leftrightarrow|j\rangle_B|k\rangle_{B'}$, the matrix elements of
$G_\mu$ are
\begin{equation}
(G_\mu)_{(j,k),i}
=
\delta_{ji}\delta_{k\mu}
+\frac12\delta_{j\mu}\delta_{ki}
-\frac12\delta_{\mu i}\delta_{jk}.
\label{eq:G-components}
\end{equation}
Substitution of Eq.~\eqref{eq:channel}, Eq.~\eqref{eq:complement-block}, and
Eq.~\eqref{eq:J-lift} into Eq.~\eqref{eq:defect} gives
\begin{align}
[\Delta(\tau)]_{(j,k),(p,q)}
=\frac29\Bigl[{}
&\delta_{kq}\tau_{jp}
+\frac12\delta_{kp}\tau_{jq}
-\frac12\delta_{pq}\tau_{jk}
+\frac12\delta_{jq}\tau_{kp}\notag\\
&+\frac14\delta_{jp}\tau_{kq}
-\frac14\delta_{pq}\tau_{kj}
-\frac12\delta_{jk}\tau_{qp}
-\frac14\delta_{jk}\tau_{pq}\notag\\
&+\frac14\delta_{jk}\delta_{pq}\Tr \tau
\Bigr].
\label{eq:defect-components}
\end{align}
On the other hand,
\begin{equation}
\left[
\sum_{\mu=0}^2G_\mu \tau G_\mu^{\ad}
\right]_{(j,k),(p,q)}
=
\sum_{\mu,i,\ell}
(G_\mu)_{(j,k),i}
\tau_{i\ell}
\overline{(G_\mu)_{(p,q),\ell}}.
\label{eq:GSG}
\end{equation}
Although the entries in Eq.~\eqref{eq:G-components} are real, the
conjugate records the adjoint in the general component formula.
For transparency, write the three summands in
Eq.~\eqref{eq:G-components} as
\[
a^\mu_{(j,k),i}=\delta_{ji}\delta_{k\mu},
\qquad
b^\mu_{(j,k),i}=\frac12\delta_{j\mu}\delta_{ki},
\qquad
c^\mu_{(j,k),i}=-\frac12\delta_{\mu i}\delta_{jk}.
\]
The row label in the following table specifies the summand to the left
of $\tau$, and the column label specifies the summand to its right.
Carrying out the sums over $\mu,i,\ell$ in
Eq.~\eqref{eq:GSG} gives
\begin{equation}
\begin{array}{c|ccc}
    & a & b & c\\ \hline
    a&
    \delta_{kq}\tau_{jp}&
    \frac12\delta_{kp}\tau_{jq}&
    -\frac12\delta_{pq}\tau_{jk}\\[1mm]
    b&
    \frac12\delta_{jq}\tau_{kp}&
    \frac14\delta_{jp}\tau_{kq}&
    -\frac14\delta_{pq}\tau_{kj}\\[1mm]
    c&
    -\frac12\delta_{jk}\tau_{qp}&
    -\frac14\delta_{jk}\tau_{pq}&
    \frac14\delta_{jk}\delta_{pq}\Tr\tau
\end{array}.
\label{eq:defect-nine-contractions}
\end{equation}
Reading the table row by row reproduces exactly the bracket in
Eq.~\eqref{eq:defect-components}.  This proves the identity in
Eq.~\eqref{eq:cp-defect}.

The Choi operator of the defect map is also particularly simple.  One checks
\begin{equation}
\Tr(G_\mu^{\ad}G_\nu)=4\delta_{\mu\nu}.
\label{eq:G-orthogonality}
\end{equation}
Thus the vectors
$|g_\mu\rangle=\frac12\operatorname{vec}(G_\mu)$ are orthonormal.
Using the unnormalized Choi convention
$C_\Delta:=\sum_{i,j=0}^2\ketbra ij\otimes\Delta(\ketbra ij)$,
and vectorizing $G_\mu$ with the input index first, one has
\begin{equation}
C_\Delta
=
\frac89\sum_{\mu=0}^2|g_\mu\rangle\!\langle g_\mu|
=
\frac89P_{\mathcal G},
\label{eq:defect-choi}
\end{equation}
where $\mathcal G=\Span\{\operatorname{vec}(G_0),
\operatorname{vec}(G_1),\operatorname{vec}(G_2)\}$.

\section{Quantum activation with \ensuremath{\Lambda}}
\label{sec:activation}

We prove that the fixed qutrit channel $\Lambda$ enables positive quantum
capacity when paired with any finite-dimensional channel of positive private
capacity. We then apply this result to an explicit four-dimensional PPT
channel.

\begin{theorem}[Universal activation by $\Lambda$]
\normalfont\par\noindent
  \label{thm:universal}
  For every finite-dimensional quantum channel $\cM$,
  \begin{equation}
    P(\cM)>0
    \quad\Longrightarrow\quad
    Q(\cM\otimes\Lambda)>0.
    \label{eq:universal}
  \end{equation}
\end{theorem}

The proof begins with a qubit restriction of $\Lambda$ and a
single-excitation encoding that produces a dephased-erasure channel.
We then use this channel to convert private codes into inputs with
positive coherent information. The final subsection applies the result
to an explicit PPT channel with positive private information.

\subsection{An exact qubit restriction of \ensuremath{\Lambda}}
\label{subsec:activation-qubit}

Put
\begin{equation}
  \ket g=\frac{\ket0+i\ket1}{\sqrt2},
  \qquad
  \ket k=\overline{\ket g}
  =\frac{\ket0-i\ket1}{\sqrt2}.
  \label{eq:g-k}
\end{equation}
The vectors $\ket g,\ket k,\ket2$ form an orthonormal basis of $\mathbb C^3$.
Define a qubit-to-qutrit isometry $V$ by
\begin{equation}
  V\ket0=\ket g,
  \qquad
  V\ket1=\ket2,
  \label{eq:encoder}
\end{equation}
and write $\cV(X)=VXV^{\ad}$.  Define a qutrit-to-qubit channel $\cR$ by
the two Kraus operators
\begin{equation}
  R_0=\ket0\bra2-\ket1\bra k,
  \qquad
  R_1=\ket0\bra g.
  \label{eq:decoder-kraus}
\end{equation}
Indeed,
\begin{equation}
  R_0^{\ad}R_0+R_1^{\ad}R_1
  =\ketbra22+\ketbra kk+\ketbra gg
  =\idop_3.
  \label{eq:decoder-tp}
\end{equation}

\begin{proposition}[Exact damped-dephasing restriction]
\normalfont\par\noindent
  \label{prop:qubit-restriction}
  The qubit channel
  \begin{equation}
    \cC:=\cR\circ\Lambda\circ\cV
    \label{eq:C-definition}
  \end{equation}
  acts on the matrix units $E_{ij}=\ketbra ij$ as
  \begin{align}
    \cC(E_{00})&=E_{00},
    \label{eq:C00}\\
    \cC(E_{11})&=\frac34E_{00}+\frac14E_{11},
    \label{eq:C11}\\
    \cC(E_{01})&=\frac14E_{01},
    \qquad
    \cC(E_{10})=\frac14E_{10}.
    \label{eq:C01}
  \end{align}
\end{proposition}

\begin{proof}
It suffices to evaluate $\Lambda$ on the two encoded basis states and
on their coherence. Since $\overline{\ket g}=\ket k$, Eq.~\eqref{eq:channel}
gives
\begin{align}
  \Lambda(\ketbra gg)
  &=\frac34\ketbra gg+\frac14\ketbra22,
  \label{eq:lambda-gg}\\
  \Lambda(\ketbra22)
  &=\frac12\ketbra22+\frac14\ketbra gg+\frac14\ketbra kk,
  \label{eq:lambda-22}\\
  \Lambda(\ketbra g2)
  &=\frac12\ketbra g2-\frac14\ketbra{2}{k}.
  \label{eq:lambda-g2}
\end{align}
The decoder maps both $\ketbra gg$ and $\ketbra22$ to $E_{00}$,
and maps $\ketbra kk$ to $E_{11}$. It also annihilates $\ketbra g2$
and sends $\ketbra{2}{k}$ to $-E_{01}$. Applying these identities to
the three outputs above proves Eqs.~\eqref{eq:C00}--\eqref{eq:C01};
the action on $E_{10}$ follows by adjoint preservation.
\end{proof}

Thus, the excitation survival probability and the ground--excitation
coherence factor are $q=c=1/4$. Under the encoding below, populations
acquire a factor $q$, whereas coherences between distinct excitation
locations acquire a factor $c^2$.

\subsection{Single-excitation encoding and dephased erasure}
\label{subsec:activation-excitation}

For an integer $d\geq1$, let $\{\ket j\}_{j=1}^d$ be the standard basis of
$\mathbb C^d$.  Define the single-excitation isometry
\begin{equation}
  W_d\ket j=\ket{e_j}
  :=
  \ket0^{\otimes(j-1)}\otimes\ket1\otimes
  \ket0^{\otimes(d-j)}.
  \label{eq:single-excitation}
\end{equation}
Let $\ket{\mathrm{vac}}=\ket0^{\otimes d}$.

\begin{proposition}[Exact dephased-erasure channel]
\normalfont\par\noindent
  \label{prop:single-excitation}
  For every input state supported on $W_d\mathbb C^d$, the output of
  $\cC^{\otimes d}$ is supported on
  \begin{equation}
    \operatorname{span}\{\ket{e_1},\ldots,\ket{e_d}\}
    \oplus\operatorname{span}\{\ket{\mathrm{vac}}\}.
  \end{equation}
  After identifying the first summand with $\mathbb C^d$ and the vacuum
  with an erasure flag $\ket{\mathsf e}$, the induced channel is
  \begin{equation}
    \cT_{q,\eta}^{(d)}(\rho)
    =q\,\cZ_\eta(\rho)
      \oplus(1-q)\Tr(\rho)\ketbra{\mathsf e}{\mathsf e},
    \label{eq:T}
  \end{equation}
  where
  \begin{equation}
    q=\frac14,
    \qquad
    \eta=\frac14,
    \qquad
    \cZ_\eta=\eta\,\idchan+(1-\eta)\Delta,
    \qquad
    \Delta(\rho)=\sum_{j=1}^d\ketbra jj\rho\ketbra jj.
    \label{eq:Z}
  \end{equation}
\end{proposition}

\begin{proof}
The ground state is fixed by $\cC$, and the sole excitation in
$\ket{e_j}$ survives with probability $q$. Hence
\begin{equation}
  \cC^{\otimes d}(\ketbra{e_j}{e_j})
  =q\ketbra{e_j}{e_j}
   +(1-q)\ketbra{\mathrm{vac}}{\mathrm{vac}}.
  \label{eq:single-diagonal}
\end{equation}
For $i\neq j$, the operator $\ketbra{e_i}{e_j}$ has one factor
$E_{10}$ and one factor $E_{01}$; all other factors are $E_{00}$.
Equation~\eqref{eq:C01} therefore gives
\begin{equation}
  \cC^{\otimes d}(\ketbra{e_i}{e_j})
  =c^2\ketbra{e_i}{e_j}.
  \label{eq:single-offdiag}
\end{equation}
These matrix units span the operators on the code space, so the two
identities determine the induced channel. Its transmitted block is
$q\Delta(\rho)+c^2[\rho-\Delta(\rho)]=q\cZ_{c^2/q}(\rho)$, and the
remaining trace lies in the orthogonal vacuum state. Since
$c^2/q=1/4$, this is Eq.~\eqref{eq:T}. For $d=1$, there are no
off-diagonal matrix units and $\Delta=\idchan$, so the same formula holds.
\end{proof}

The vacuum flags erasure. Within the transmitted block, however, the
receiver is not told whether the identity or dephasing component of
$\cZ_{1/4}$ occurred.

\subsection{From private information to quantum activation}
\label{subsec:activation-conversion}

We first derive a coherent-information lower bound for an arbitrary input
ensemble. We then apply it to a sequence of private codes to prove
Theorem~\ref{thm:universal}.

Let $\cM:A\to B$ be a channel, let $U^{\cM}_{A\to BE}$ be an isometric
extension, and let $\{p_x,\rho_A^x\}_{x\in\mathcal X}$ be a finite
ensemble, with any zero-probability terms omitted. Define the classical--quantum
(cq) states
\begin{align}
  \omega_{XB}
  &=\sum_xp_x\ketbra xx_X\otimes\cM(\rho_A^x),
  \label{eq:cq-B}\\
  \omega_{XE}
  &=\sum_xp_x\ketbra xx_X\otimes\cM^c(\rho_A^x).
  \label{eq:cq-E}
\end{align}
Put
\begin{equation}
  a:= H(X|E)_\omega=H(X)-I(X;E)_\omega,
  \qquad
  b:= H(X|B)_\omega=H(X)-I(X;B)_\omega.
  \label{eq:a-b}
\end{equation}

\begin{samepage}
\begin{lemma}[Private-to-coherent-information conversion]
\normalfont\par\noindent
  \label{lem:private-to-coherent}
  For $0\leq q,\eta\leq1$, let $\cT_{q,\eta}^{(d)}$ be defined by
  Eq.~\eqref{eq:T}, with $\cZ_\eta=\eta\idchan+(1-\eta)\Delta$, and set
  \begin{equation}
    d:=\sum_{x\in\mathcal X}\rank\rho_A^x.
    \label{eq:d-definition}
  \end{equation}
  There is an input to $\cM\otimes\cT_{q,\eta}^{(d)}$ whose coherent
  information satisfies
  \begin{equation}
    \Ic\geq q\eta a-qh_2(\eta)-(1-q)b.
    \label{eq:conversion-bound}
  \end{equation}
\end{lemma}
\end{samepage}

\begin{proof}
We construct the input by purifying each signal state in an auxiliary
subspace labeled by its message. Choose spectral decompositions
$\rho_A^x=\sum_{r=1}^{r_x}\lambda_{r|x}\ketbra{\psi_{xr}}{\psi_{xr}}$
with $\lambda_{r|x}>0$ and $r_x=\rank\rho_A^x$. Give
$C\simeq\mathbb C^d$ the orthonormal basis $\{\ket{x,r}\}_{x,r}$, and set
\begin{equation}
  \ket{\phi_x}_{AC}
  :=\sum_{r=1}^{r_x}\sqrt{\lambda_{r|x}}\ket{\psi_{xr}}_A\ket{x,r}_C,
  \qquad
  \rho_{AC}:=\sum_xp_x\ketbra{\phi_x}{\phi_x}.
  \label{eq:conditional-purification}
\end{equation}
Use $\rho_{AC}$ as the channel input, with purification
\begin{equation}
  \ket\Phi_{XAC}:=\sum_x\sqrt{p_x}\ket x_X\ket{\phi_x}_{AC}.
  \label{eq:global-purification}
\end{equation}
Here, $X$ is a quantum reference, not a classical register in the full
purification. The dephasing map $\Delta$ acts in the basis $\{\ket{x,r}\}$.

Because the erasure flag is orthogonal to the transmitted subspace,
the coherent information is the weighted sum of the two branch
contributions:
\begin{equation}
  \Ic=qI_{\mathrm s}+(1-q)I_{\mathrm e}.
  \label{eq:flagged-average}
\end{equation}
Indeed, the direct-sum entropy term $h_2(q)$ cancels between the receiver
entropy and the joint reference--receiver entropy. In the erased branch,
tracing out $C$ removes all cross terms between distinct pairs $(x,r)$.
The remaining state on $XB$ is exactly $\omega_{XB}$, so
\begin{equation}
  I_{\mathrm e}=-H(X|B)_\omega=-b.
  \label{eq:erased-branch}
\end{equation}

For the transmitted branch, write
$\sigma^{\mathrm s}=\eta\sigma^{\mathrm{id}}+(1-\eta)\sigma^\Delta$ on
$XBC$, where the two component states result from applying $\cM$ to $A$
and, respectively, $\idchan$ or $\Delta$ to $C$.
For the identity component, the state on $XBEC$ after $U^{\cM}$ is pure,
and its $XE$ marginal is $\omega_{XE}$ because the auxiliary subspaces
for different $x$ are orthogonal. Thus
\begin{equation}
  I_{\mathrm{id}}
  :=S(BC)_{\sigma^{\mathrm{id}}}-S(XBC)_{\sigma^{\mathrm{id}}}
  =S(XE)_\omega-S(E)_\omega=a.
  \label{eq:identity-contribution}
\end{equation}
For the dephased component,
\begin{equation}
  \sigma^\Delta_{XBC}
  =\sum_{x,r}p_x\lambda_{r|x}\ketbra xx_X
   \otimes\cM(\ketbra{\psi_{xr}}{\psi_{xr}})_B
   \otimes\ketbra{x,r}{x,r}_C.
  \label{eq:delta-with-X}
\end{equation}
Since $C$ already records $x$, adjoining $X$ does not change the entropy:
$I_\Delta:=S(BC)_{\sigma^\Delta}-S(XBC)_{\sigma^\Delta}=0$.
Entropy concavity on $BC$ and the binary-mixture entropy upper bound
on $XBC$ now give
\begin{equation}
\begin{aligned}
  I_{\mathrm s}
  &=S(BC)_{\sigma^{\mathrm s}}-S(XBC)_{\sigma^{\mathrm s}}\\
  &\geq\eta I_{\mathrm{id}}+(1-\eta)I_\Delta-h_2(\eta)
   =\eta a-h_2(\eta).
\end{aligned}
\label{eq:success-bound}
\end{equation}
Combining this with Eqs.~\eqref{eq:flagged-average} and
\eqref{eq:erased-branch} proves the claimed bound. A branch with zero
probability contributes zero, so the argument also covers $q=0,1$;
the mixture bound includes $\eta=0,1$.
\end{proof}

\begin{remark}[Why the entropy penalty is dimension independent]
\normalfont\par\noindent
The penalty $h_2(\eta)\leq1$ comes from mixing the two complete channel
components $\idchan$ and $\Delta$, not from resolving the $d$ dephasing
outcomes. Neither the component label nor those outcomes are supplied to
the receiver, and the bound is independent of $d$.
\end{remark}

We now apply Lemma~\ref{lem:private-to-coherent} to private codes for
$\cM^{\otimes n}$. The secrecy and reliability estimates make the
coherent-information lower bound positive at a finite blocklength.

\begin{proof}[Proof of Theorem~\ref{thm:universal}]
Since $P(\cM)>0$, the private coding theorem~\cite{Devetak2005} provides
a positive-rate sequence of codes for $\cM^{\otimes n}$ with $K_n$
equiprobable messages, average decoding error $\varepsilon_n\to0$, and
Holevo leakage $I(X;E^n)\to0$. In particular, $\log K_n\to\infty$.
Discarding finitely many initial codes if necessary, assume $K_n\geq2$.
Regard the codewords as an ensemble
$\{K_n^{-1},\rho_x^{(n)}\}_{x=1}^{K_n}$. Secrecy gives
\begin{equation}
  a_n:=H(X|E^n)=\log K_n-o(1).
  \label{eq:an}
\end{equation}
If $\widehat X$ is the decoder output, data processing and Fano's
inequality give
\begin{equation}
\begin{aligned}
  b_n:=H(X|B^n)
  &\leq H(X|\widehat X)\\
  &\leq h_2(\varepsilon_n)+\varepsilon_n\log(K_n-1)
   =o(\log K_n).
\end{aligned}
\label{eq:bn}
\end{equation}
Apply Lemma~\ref{lem:private-to-coherent} to this ensemble, with
$q=\eta=1/4$ and $d_n:=\sum_x\rank\rho_x^{(n)}$. It gives an input to
$\cM^{\otimes n}\otimes\cT_{1/4,1/4}^{(d_n)}$ whose coherent information
satisfies
\begin{equation}
\begin{aligned}
  I_{\mathrm c}^{(n)}
  &\geq\frac1{16}a_n-\frac34b_n-\frac14h_2(1/4)\\
  &=\left(\frac1{16}-o(1)\right)\log K_n-\frac14h_2(1/4)>0
\end{aligned}
\label{eq:main-lower-bound}
\end{equation}
for all sufficiently large $n$.

Fix one such finite $n$. Propositions~\ref{prop:qubit-restriction}
and~\ref{prop:single-excitation} realize
$\cT_{1/4,1/4}^{(d_n)}$ by an isometric input encoding into
$\Lambda^{\otimes d_n}$ followed by receiver postprocessing. Isometric
encoding preserves the input purification, and receiver postprocessing
cannot increase coherent information. Therefore the encoded input
already has positive coherent information for
$\cM^{\otimes n}\otimes\Lambda^{\otimes d_n}$.
To embed this construction into equal numbers of channel uses, set
$L=\max\{n,d_n\}$ and feed fixed pure states into the unused factors.
These factors contribute zero coherent information, so additivity on
product inputs preserves the positive value. The quantum coding theorem
then gives
\begin{equation}
  Q(\cM\otimes\Lambda)
  \geq\frac1L Q^{(1)}\!\left((\cM\otimes\Lambda)^{\otimes L}\right)>0.
  \label{eq:positive-Q}
\end{equation}
\end{proof}

\subsection{An explicit PPT activating partner}
\label{subsec:ppt-activating-partner}

We now specify a fixed channel $\cM_H$ to which
Theorem~\ref{thm:universal} applies.  It is the channel associated with a $4\otimes4$ PPT state having a positive one-way secret key, originating in the low-dimensional Horodecki construction~\cite{HorodeckiEtAl2008} and
used in the superactivation framework of Smith and Yard~\cite{SmithYard2008}.
We specify the normalized Choi state and tensor ordering explicitly.

\paragraph{Construction and PPT property.}
Let $A,B,A',B'$ be qubits.  On $AB$, define
\begin{align}
  \ket{\beta_0}&=\frac{\ket{00}+\ket{11}}{\sqrt2},
  &
  \ket{\beta_1}&=\frac{\ket{00}-\ket{11}}{\sqrt2},
  \label{eq:bell-phi}\\
  \ket{\beta_2}&=\frac{\ket{01}+\ket{10}}{\sqrt2},
  &
  \ket{\beta_3}&=\frac{\ket{01}-\ket{10}}{\sqrt2}.
  \label{eq:bell-psi}
\end{align}
On the shield $A'B'$, put
\begin{align}
  \tau_0
  &=\frac12\left(\ketbra{00}{00}+\ketbra{\beta_2}{\beta_2}\right),
  \label{eq:tau0}\\
  \tau_1
  &=\frac12\left(\ketbra{11}{11}+\ketbra{\beta_3}{\beta_3}\right),
  \label{eq:tau1}\\
  \tau_2&=\ketbra{\chi_+}{\chi_+},
  \qquad
  \tau_3=\ketbra{\chi_-}{\chi_-},
  \label{eq:tau23}
\end{align}
where the Bell vectors in Eqs.~\eqref{eq:tau0}--\eqref{eq:tau1} are now
understood on $A'B'$, and
\begin{align}
  \ket{\chi_+}
  &=\frac12\left(\sqrt{2+\sqrt2}\ket{00}
                 +\sqrt{2-\sqrt2}\ket{11}\right),
  \label{eq:chi-plus}\\
  \ket{\chi_-}
  &=\frac12\left(\sqrt{2-\sqrt2}\ket{00}
                 -\sqrt{2+\sqrt2}\ket{11}\right).
  \label{eq:chi-minus}
\end{align}
Set
\begin{equation}
  q_H=2-\sqrt2
  \label{eq:qH}
\end{equation}
and define, initially in the tensor order $(AB)(A'B')$,
\begin{equation}
  \Omega_H
  =\frac{q_H}{2}\sum_{j=0}^{1}
    \ketbra{\beta_j}{\beta_j}_{AB}\otimes\tau_j
   +\frac{1-q_H}{2}\sum_{j=2}^{3}
    \ketbra{\beta_j}{\beta_j}_{AB}\otimes\tau_j.
  \label{eq:OmegaH}
\end{equation}
For the Choi bipartition, regroup the systems as $AA':BB'$. We write
$\Gamma_{BB'}$ for partial transposition on the receiver systems $BB'$.

\begin{proposition}[PPT and trace-preserving Choi conditions]
\normalfont\par\noindent
  \label{prop:helper-choi}
  The state $\Omega_H$ satisfies
  \begin{equation}
    \Omega_H\succeq0,
    \qquad
    \Omega_H^{\GammaB}\succeq0,
    \qquad
    \Tr_{BB'}\Omega_H=\frac14\idop_{AA'}.
    \label{eq:helper-properties}
  \end{equation}
  Consequently $J_H=4\Omega_H$ is the unnormalized Choi matrix of a
  trace-preserving PPT channel
  $\cM_H:AA'\to BB'$.
\end{proposition}

\begin{proof}
The Bell--shield mixture in Eq.~\eqref{eq:OmegaH} makes positivity and
normalization immediate. For the remaining properties, it is useful to
pair the shield states:
\begin{equation}
\begin{aligned}
  S_0&:=\tau_0+\tau_1=\tfrac12\idop_{A'B'},
  &D_0&:=\tau_0-\tau_1,\\
  S_1&:=\tau_2+\tau_3=\ketbra{00}{00}+\ketbra{11}{11},
  &D_1&:=\tau_2-\tau_3.
\end{aligned}
\label{eq:helper-shield-pairs}
\end{equation}
Both $S_0$ and $S_1$ have $B'$ partial trace $\idop_{A'}$.
Each Bell state has $B$ partial trace $\idop_A/2$, so
\begin{equation}
  \Tr_{BB'}\Omega_H
  =\frac14\idop_A\otimes
   \Tr_{B'}\!\left[q_HS_0+(1-q_H)S_1\right]
  =\frac14\idop_{AA'}.
  \label{eq:helper-marginal-direct}
\end{equation}

In fact, $\Omega_H$ is invariant under partial transposition.
In the key basis $(\ket{00},\ket{01},\ket{10},\ket{11})_{AB}$,
Eq.~\eqref{eq:OmegaH} reads
\begin{equation}
\Omega_H=\frac14
\begin{pmatrix}
q_HS_0&0&0&q_HD_0\\
0&(1-q_H)S_1&(1-q_H)D_1&0\\
0&(1-q_H)D_1&(1-q_H)S_1&0\\
q_HD_0&0&0&q_HS_0
\end{pmatrix},
\label{eq:helper-key-blocks}
\end{equation}
where each block acts on $A'B'$. The diagonal blocks are unchanged by
$T_{B'}$. For the off-diagonal blocks, the shield definitions give
\begin{equation}
\begin{aligned}
D_0^{T_{B'}}
 &=\tfrac12\left(\ketbra{00}{00}-\ketbra{11}{11}
                 +\ketbra{00}{11}+\ketbra{11}{00}\right)
   =D_1/\sqrt2,\\
D_1^{T_{B'}}&=\sqrt2D_0.
\end{aligned}
\label{eq:helper-shield-transpose}
\end{equation}
Partial transposition on $B$ interchanges the corner off-diagonal blocks
with the two central off-diagonal blocks. Applying $T_{B'}$ to these
blocks restores the original matrix because
$q_H=\sqrt2(1-q_H)$. Hence
\begin{equation}
  \Omega_H^{\GammaB}=\Omega_H\succeq0.
  \label{eq:helper-ppt-invariance}
\end{equation}
Together with the marginal identity, this shows that $J_H=4\Omega_H$
defines a trace-preserving completely positive map by
\begin{equation}
  \cM_H(X)
  =\Tr_{AA'}\!\left[(X^{\mathsf T}\otimes\idop_{BB'})J_H\right].
  \label{eq:choi-reconstruction}
\end{equation}
Its Choi matrix is PPT, or equivalently
$\mathsf T_{BB'}\circ\cM_H$ is completely positive.
\end{proof}

\paragraph{Positive private information.}
We evaluate the private information of a binary input ensemble.

\begin{proposition}[Positive private information]
\normalfont\par\noindent
  \label{prop:helper-private}
  For the equiprobable input ensemble
  \begin{equation}
    \rho_x=\ketbra xx_A\otimes\frac12\idop_{A'},
    \qquad x\in\{0,1\},
    \label{eq:helper-ensemble}
  \end{equation}
  the channel $\cM_H$ satisfies
  \begin{equation}
    I(X;BB')-I(X;E_H)
    =1-h_2(2-\sqrt2)
    =0.021339915649840\ldots>0,
    \label{eq:helper-private-info}
  \end{equation}
  where $E_H$ is the output of a complementary channel.
\end{proposition}

\begin{proof}
The receiver's output can be read directly from the diagonal key blocks
in Eq.~\eqref{eq:helper-key-blocks}. Since
$\Tr_{A'}S_0=\Tr_{A'}S_1=\idop_{B'}$, the Choi reconstruction formula gives
\begin{equation}
  \cM_H(\rho_x)
  =\left[q_H\ketbra xx+(1-q_H)\ketbra{1-x}{1-x}\right]_B
   \otimes\tfrac12\idop_{B'}.
  \label{eq:helper-receiver-output}
\end{equation}
For equiprobable inputs, the average output is $\idop_{BB'}/4$ and
each conditional output has entropy $1+h_2(q_H)$. Therefore
\begin{equation}
  I(X;BB')=1-h_2(q_H).
  \label{eq:receiver-holevo}
\end{equation}

To see why the environment learns nothing about $x$, let
$Z_A:=\ketbra00_A-\ketbra11_A$. The shield supports satisfy
$\tau_0\tau_1=0$ and $\tau_2\tau_3=0$, while $Z_A$ exchanges
$\ket{\beta_0}$ with $\ket{\beta_1}$ and $\ket{\beta_2}$ with
$\ket{\beta_3}$. Thus $Z_A$ maps the support of $\Omega_H$ into its
orthogonal complement. If $\Pi_H$ is the support projector, this says
\begin{equation}
  \Pi_H\left(Z_A\otimes\idop_{A'BB'}\right)\Pi_H=0,
  \label{eq:helper-support-orthogonality}
\end{equation}
where the tensor order is $AA':BB'$.

Take an orthonormal eigendecomposition
$\Omega_H=\sum_{\ell=1}^6\lambda_\ell\ketbra{v_\ell}{v_\ell}$
on its support, with input-first coordinates
$\ket{v_\ell}=\sum_{i,j}(v_\ell)_{ij}\ket i_{AA'}\ket j_{BB'}$.
The corresponding Kraus operators obey
$(K_\ell)_{ji}=2\sqrt{\lambda_\ell}(v_\ell)_{ij}$. Since
$\rho_x=\tfrac14[\idop_{AA'}+(-1)^xZ_A\otimes\idop_{A'}]$, their full
complementary channel satisfies
\begin{equation}
\begin{aligned}
 {}[\cM_H^c(\rho_x)]_{\ell m}
 &=\Tr(K_\ell\rho_xK_m^\dagger)\\
 &=4\sqrt{\lambda_\ell\lambda_m}
   \langle v_m|(\rho_x^{\mathsf T}\otimes\idop_{BB'})|v_\ell\rangle
  =\lambda_\ell\delta_{\ell m}.
\end{aligned}
\label{eq:complement-gram}
\end{equation}
The last equality uses Eq.~\eqref{eq:helper-support-orthogonality}.
The Bell blocks in Eq.~\eqref{eq:OmegaH} have nonzero eigenvalues
$q_H/4$ four times and $(1-q_H)/2$ twice. Ordering the eigenvectors
accordingly, we obtain the common environment state
\begin{equation}
 \cM_H^c(\rho_0)=\cM_H^c(\rho_1)
 =\operatorname{diag}\!\left(
  \tfrac{q_H}{4},\tfrac{q_H}{4},\tfrac{q_H}{4},\tfrac{q_H}{4},
  \tfrac{1-q_H}{2},\tfrac{1-q_H}{2}\right).
 \label{eq:environment-equality}
\end{equation}
Thus $I(X;E_H)=0$, and Eq.~\eqref{eq:receiver-holevo} gives
Eq.~\eqref{eq:helper-private-info}. Its value is positive because
$q_H\neq1/2$.
\end{proof}

\paragraph{PPT superactivation.}
Combining the preceding construction with Theorem~\ref{thm:universal} gives the following result.

\begin{theorem}[PPT activation of $\Lambda$]
\normalfont\par\noindent
  \label{thm:main}
  There exists a channel
  $\cM_H:\Lin(\mathbb C^4)\to\Lin(\mathbb C^4)$ whose Choi matrix is PPT
  and for which
  \begin{equation}
    Q(\cM_H\otimes\Lambda)>0.
    \label{eq:main-result}
  \end{equation}
\end{theorem}

\begin{proof}[Proof of Theorem~\ref{thm:main}]
Proposition~\ref{prop:helper-private} and the private coding theorem give
$P(\cM_H)\geq P^{(1)}(\cM_H)>0$. Applying
Theorem~\ref{thm:universal} therefore yields
$Q(\cM_H\otimes\Lambda)>0$.

The PPT property also ensures $Q(\cM_H)=0$: after any number of channel
uses and any receiver decoding, the reference--receiver state remains
PPT. For a decoded PPT state $\zeta$ on two $K$-dimensional systems,
\[
  \Tr(\Phi_K\zeta)
  =\frac1K\Tr(F_K\zeta^{T_B})\leq\frac1K,
\]
where $\Phi_K$ is the normalized maximally entangled projector and $F_K$
is the swap. The inequality follows from $F_K\preceq\idop$ and
$\zeta^{T_B}\succeq0$, so these states cannot approach a maximally
entangled state at a positive communication rate. Together with
$Q(\Lambda)=0$, this proves superactivation between two zero-quantum-capacity
channels.
\end{proof}

\section{Quantum communication from two nonprivate factors}
\label{sec:zero-private-quantum}

The complete comparison proved for $\Lambda$ is stronger than the condition
$P(\Lambda)=0$. The following stability property makes the distinction
explicit.

\begin{corollary}[Stability under complete domination]
\normalfont\par\noindent
\label{cor:complete-partner-stability}
If a finite-dimensional channel $\cA$ satisfies
$\cA^c\REord\cA$, then
\begin{equation}
 P(\cA\otimes\Lambda)=Q(\cA\otimes\Lambda)=0.
 \label{eq:complete-partner-zero}
\end{equation}
This includes every antidegradable $\cA$.
\end{corollary}

\begin{proof}
The two-pair tensorization in Eq.~\eqref{eq:hrsf-two-pair} and
$\Lambda^c\REord\Lambda$ imply
$\cA^c\otimes\Lambda^c\REord\cA\otimes\Lambda$.
The left-hand side is a complement of the product channel, up to an output
isometry. Apply Corollary~\ref{cor:anti-less-noisy-capacity}.
For antidegradable $\cA$, the premise follows directly from data processing
under every reference extension.
\end{proof}

We now combine Theorem~\ref{thm:universal} with the private-capacity
superactivation result of Zhu and Wang~\cite{ZhuWang2026Private}.
By Theorem~1 of Ref.~\cite{ZhuWang2026Private}, there is a fixed channel
$\cN_0:\Lin(\mathbb C^4)\to\Lin(\mathbb C^4)$ such that
\begin{equation}
 P(\cN_0)=0,\qquad
 P(\cN_0\otimes\mathcal E_{2,1/2})>0.
 \label{eq:imported-private-superactivation}
\end{equation}
Here $\cN_0$ is the channel defined in Eq.~(2.2) of
Ref.~\cite{ZhuWang2026Private}, and
\begin{equation}
 \mathcal E_{2,1/2}(X)
 :=\tfrac12 X\oplus\tfrac12\Tr(X)\ketbra{\mathsf e}{\mathsf e}
 \label{eq:half-erasure-definition}
\end{equation}
is the qubit half-erasure channel. Its antidegradability, together with
Eq.~\eqref{eq:imported-private-superactivation}, yields the following corollary
of Theorem~\ref{thm:universal}.

\begin{corollary}[Quantum activation of two nonprivate factors]
\normalfont\par\noindent
\label{cor:two-nonprivate-quantum}
Let
\begin{equation}
 \cA_0:=\cN_0,\qquad
 \cB_0:=\mathcal E_{2,1/2}\otimes\Lambda.
 \label{eq:two-nonprivate-factors}
\end{equation}
Then
\begin{equation}
 P(\cA_0)=P(\cB_0)=0,\qquad Q(\cA_0\otimes\cB_0)>0.
 \label{eq:two-nonprivate-positive-Q}
\end{equation}
The input--output dimensions of $\cA_0$ and $\cB_0$ are $4\to4$ and
$6\to9$, respectively.
\end{corollary}

\begin{proof}
The equality $P(\cA_0)=0$ follows from
Eq.~\eqref{eq:imported-private-superactivation}.
Since $\mathcal E_{2,1/2}$ is antidegradable,
Corollary~\ref{cor:complete-partner-stability} gives $P(\cB_0)=0$.
Apply Theorem~\ref{thm:universal} to the finite-dimensional channel
$\cM=\cN_0\otimes\mathcal E_{2,1/2}$, which has positive private capacity
by Eq.~\eqref{eq:imported-private-superactivation}. It gives
\[
 Q(\cN_0\otimes\mathcal E_{2,1/2}\otimes\Lambda)>0.
\]
Regrouping the tensor factors yields Eq.~\eqref{eq:two-nonprivate-positive-Q}.
\end{proof}

Unlike the PPT-helper construction, in which $P(\cM_H)>0$, both factors in Corollary~\ref{cor:two-nonprivate-quantum} have zero private capacity. The
argument combines the private-capacity result of Ref.~\cite{ZhuWang2026Private} with Theorem~\ref{thm:universal}; it does not assume additivity of either capacity.

\end{document}